\documentclass[sigplan,screen]{acmart}

\newif\ifextended
\extendedtrue

\setcopyright{acmlicensed}
\acmPrice{15.00}
\acmDOI{10.1145/3314221.3314614}
\acmYear{2019}
\copyrightyear{2019}
\acmISBN{978-1-4503-6712-7/19/06}
\acmConference[PLDI '19]{Proceedings of the 40th ACM SIGPLAN Conference on Programming Language Design and Implementation}{June 22--26, 2019}{Phoenix, AZ, USA}
\acmBooktitle{Proceedings of the 40th ACM SIGPLAN Conference on Programming Language Design and Implementation (PLDI '19), June 22--26, 2019, Phoenix, AZ, USA}

\usepackage{booktabs}   
\usepackage{subcaption} 
\usepackage{amsmath}
\usepackage{amssymb}
\usepackage{pifont}
\usepackage{algorithm}
\usepackage{xspace}
\usepackage{soul, xcolor}
\usepackage{algorithmicx}
\usepackage[noend]{algpseudocode}
\usepackage{amsthm}
\usepackage{mathtools}
\usepackage{subcaption}
\usepackage{enumitem}
\usepackage{url}
\newcommand{\real}{\ensuremath{\mathbb{R}}}
\newcommand{\relu}{\ensuremath{\mathrm{ReLU}}}
\newcommand{\ai}{${\rm AI}^2$}
\newcommand{\reluplex}{{\sc Reluplex}\xspace}
\newcommand{\reluval}{{\sc ReluVal}\xspace}
\newcommand{\powerset}{\mathcal{P}}
\newcommand{\network}{\mathcal{N}}
\newcommand{\weights}{W}
\renewcommand{\vec}[1]{\overline{#1}}
\newcommand{\bias}{\vec{b}}
\newcommand{\layer}{L}
\newcommand{\transpose}{\top}
\newcommand{\ball}{\mathcal{B}}
\newcommand{\domain}{\alpha}
\newcommand{\objective}{\mathcal{F}}

\newcommand{\policy}{\pi}
\newcommand{\apolicy}{\pi^\domain}
\newcommand{\ipolicy}{\pi^I}

\algnewcommand\Input{\textbf{input: }}
\algnewcommand\Output{\textbf{output: }}

\DeclareMathOperator*{\argmin}{arg\,min}

\newcommand{\todo}[1]{{\color{red}{#1}}}

\makeatletter
\newtheorem*{rep@theorem}{\rep@title}
\newcommand{\newreptheorem}[2]{
  \newenvironment{rep#1}[1]{
    \def\rep@title{#2 \ref{##1}}
    \begin{rep@theorem}}
      {\end{rep@theorem}}}
\makeatother

\newtheorem{assumption}{Assumption}
\newreptheorem{assumption}{Assumption}
\newreptheorem{theorem}{Theorem}
\newreptheorem{definition}{Definition}

\newcommand{\toolname}{{\sc Charon}\xspace}

\begin{document}

\setstcolor{red}

\title[Optimization and Abstraction: A Synergistic Approach\dots]{Optimization and Abstraction: \\ A Synergistic Approach for Analyzing
  Neural Network Robustness}



\author{Greg Anderson}
\affiliation{
  \institution{The University of Texas at Austin}            
  \city{Austin}
  \state{Texas}
  \country{USA}                    
}
\email{ganderso@cs.utexas.edu}          

\author{Shankara Pailoor}
\affiliation{
  \institution{The University of Texas at Austin}           
  \city{Austin}
  \state{Texas}
  \country{USA}                   
}
\email{spailoor@cs.utexas.edu}         

\author{Isil Dillig}
\affiliation{
  \institution{The University of Texas at Austin}            
  \city{Austin}
  \state{Texas}
  \country{USA}                    
}
\email{isil@cs.utexas.edu}          

\author{Swarat Chaudhuri}
\affiliation{
  \institution{Rice University}            
  \city{Houston}
  \state{Texas}
  \country{USA}                    
}
\email{swarat@rice.edu}          

\begin{abstract}
  \vspace{-0.05in}
	In recent years, the notion of \emph{local robustness} (or \emph{robustness} for short) has emerged as a desirable property of deep neural networks. Intuitively, robustness means that small perturbations to an input do not cause the network to perform misclassifications. In this paper, we present a novel algorithm for verifying robustness properties of neural networks. 
Our method synergistically combines  gradient-based optimization methods for
counterexample search with abstraction-based proof search to obtain a sound and
($\delta$-)complete decision procedure.
Our method also employs a data-driven approach to learn a verification policy
that guides abstract interpretation during proof search. We have implemented the
proposed approach in a tool called \toolname\ and experimentally evaluated it on hundreds of benchmarks. Our experiments show that the proposed approach {significantly} outperforms three state-of-the-art tools, namely {\sc AI}$^2$, {\sc Reluplex}, and {\sc Reluval}.
\vspace{-0.07in}
\end{abstract}

\ccsdesc[500]{Theory of computation~Abstraction}
\ccsdesc[300]{Computing methodologies~Neural networks}

\keywords{\vspace{-0.04in}Machine learning, Abstract Interpretation, Optimization, Robustness}  

\maketitle

\section{Introduction}
\label{sec:intro}

In recent years, deep neural networks (DNNs) have gained enormous popularity for a wide spectrum of 
applications, ranging from image recognition\cite{image-recognition-1, image-recognition-2} and malware detection \cite{malware-1,malware-2} to machine translation\cite{machine-translation}. 
Due to their surprising effectiveness in practice, deep learning has also found numerous applications in safety-critical systems, including self-driving cars~\cite{car-1, car-2}, unmanned aerial systems~\cite{uas}, and medical diagnosis~\cite{diagnosis}.

Despite their widespread use in a broad range of application domains, it is well-known that deep neural networks are vulnerable to \emph{adversarial counterexamples}, which are small perturbations to a network's input that cause the network to output incorrect labels~\cite{adversarial-1,adversarial-2}. For instance, Figure~\ref{fig:adversarial} shows two adversarial examples in the context of speech recognition and image classification. As shown in the top half of Figure~\ref{fig:adversarial}, two sound waves that are virtually indistinguishable are recognized as ``How are you?" and ``Open the door"  by a DNN-based speech recognition system~\cite{gong2018overview}. Similarly, as illustrated in the bottom half of the same figure, applying a tiny perturbation to a panda image causes a  DNN to misclassify the image as that of a gibbon.

It is by now well-understood that such adversarial counterexamples can pose
serious security risks~\cite{advsurvey}.  Prior
work~\cite{BastaniILVNC16,reluplex,deepxplore,veriviz,ai2} has advocated the property of
\emph{local robustness} (or \emph{robustness} for short) for protecting neural
networks against attacks that exploit such adversarial examples. To understand
what robustness means, consider a neural network that classifies an input $x$ as
having label $y$. Then, local robustness requires that all inputs $x'$ that are
``very similar''~\footnote{For example, ``very similar'' may mean $x'$ is within some $\epsilon$ distance from $x$, where distance can be measured using different metrics such as $L^2$ norm.} to $x$ are also classified as having the same label $y$.


\begin{figure}[t]
\includegraphics[width=0.8\columnwidth]{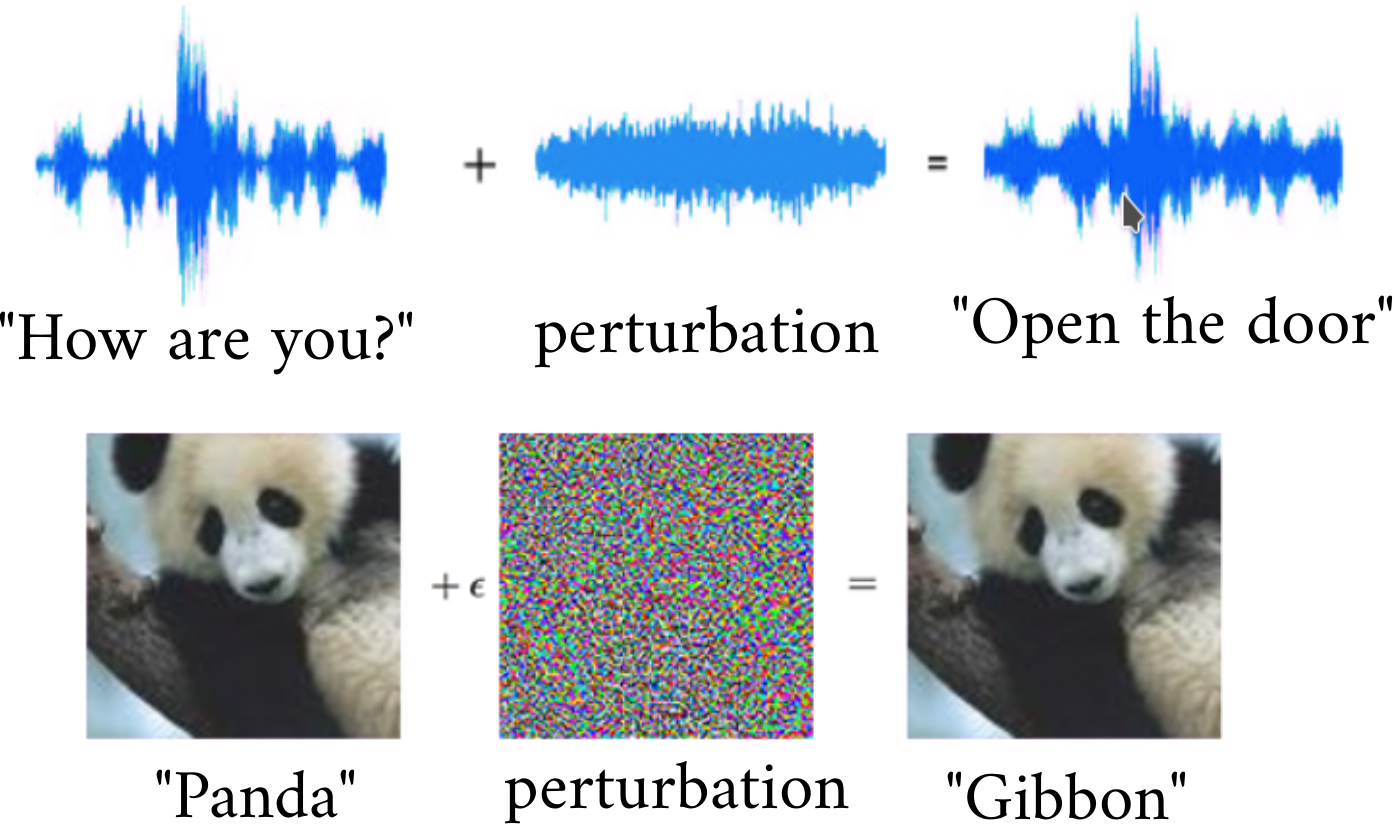}
\vspace{-0.1in}
\caption{ Small perturbations of the input cause the sound wave and the image to be misclassified.} \label{fig:adversarial}
\vspace{-0.15in}
\end{figure}

Due to the growing consensus on the desirability of robust neural networks, many recent efforts have sought to algorithmically analyze robustness of networks. Of these, one category of methods seeks to discover {adversarial counterexamples} using numerical optimization techniques such as Projected Gradient Descent (PGD)~\cite{madry2017towards} and Fast Gradient Sign Method (FGSM)~\cite{GoodfellowSS14}. A second category  aims to \emph{prove}  network robustness using symbolic methods ranging from SMT-solving~\cite{reluplex} to abstract interpretation~\cite{ai2,reluval}. 
These two categories of methods  have complementary advantages. Numerical
counterexample search methods can quickly find violations, but are ``unsound'',
in that they fail to offer certainty about a network's robustness. In contrast,
proof search methods are sound, but they are either incomplete~\cite{ai2} (i.e., suffer from false positives) or do not scale well~\cite{reluplex}.

In this paper, we present a new technique for robustness analysis of neural networks that combines the best of proof-based and optimization-based methods. Our approach combines formal reasoning techniques based on \emph{abstract interpretation} with continuous and black-box optimization techniques from the machine learning community. This tight coupling of optimization and abstraction has two key advantages: First, optimization-based methods can efficiently search for counterexamples that prove the violation of the robustness property, allowing efficient falsification in addition to verification. 
Second, 
optimization-based methods provide a data-driven way to automatically refine the
abstraction when the property can be neither falsified nor proven. 


The workflow of our approach is shown schematically in Figure~\ref{fig:cegar} and consists of 
both a \emph{training} and a \emph{deployment} phase. During the
training phase, our method uses black-box optimization
techniques to learn a so-called \emph{verification policy} $\policy_\theta$ from
a representative set of training problems. Then
the deployment phase
uses the learned verification policy to guide how gradient-based counterexample search should be coupled with proof synthesis for solving  previously-unseen verification problems.

\begin{figure}[h]
  \includegraphics[width=0.9\columnwidth]{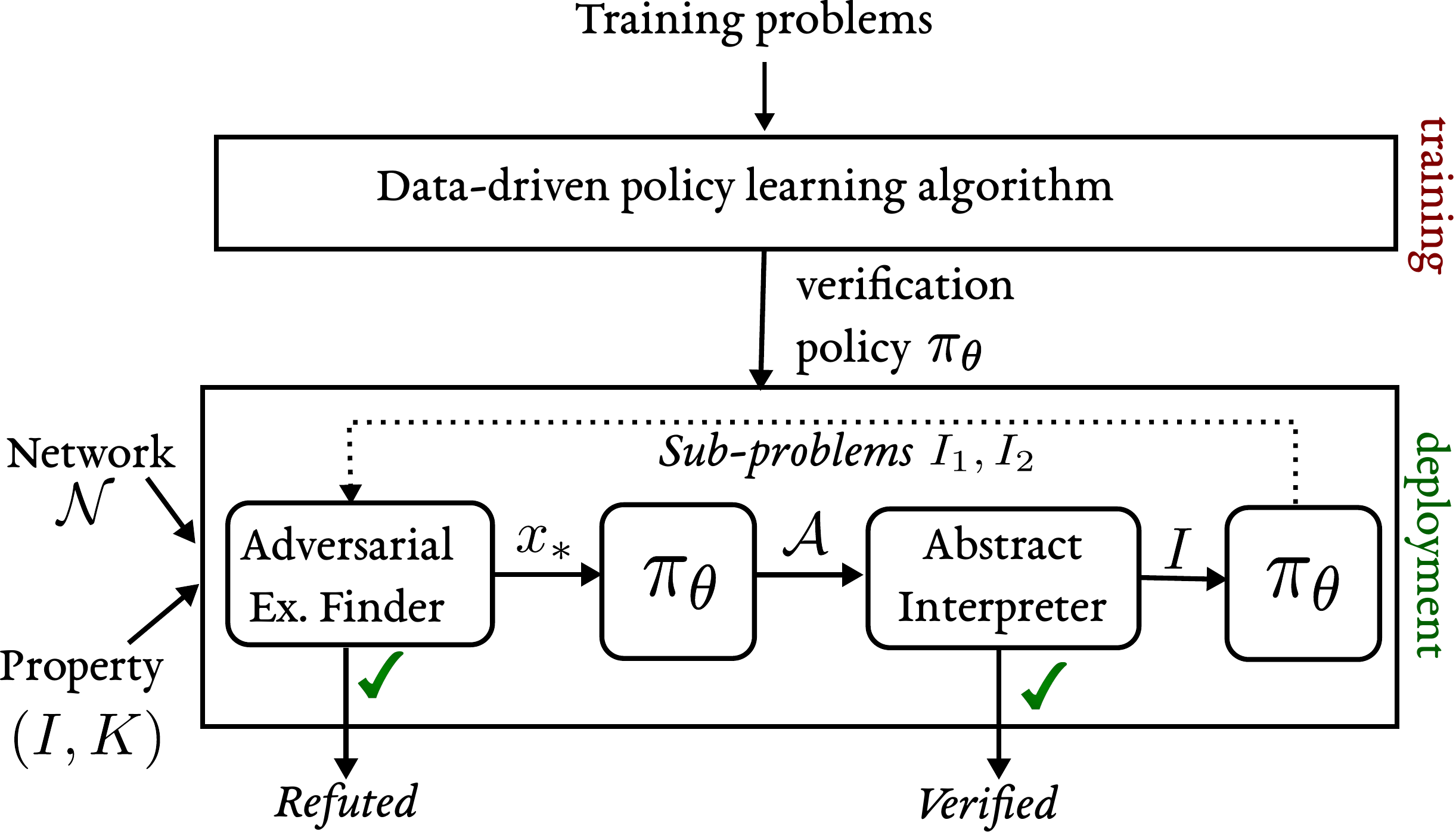}
  \vspace{-0.1in}
  \caption{Schematic overview of our approach.}
  \label{fig:cegar}
  \vspace{-0.15in}
\end{figure}

The input to the deployment phase of our algorithm consists of a neural network $\network$ as well as a robustness specification $(I, K)$ which states that all points in the \emph{input region} $I$ should be classified as having label $K$. Given this input, our algorithm first uses gradient-based optimization to search for an adversarial counterexample, which is a point in the input region $I$ that is classified as having  label $K' \neq K$. 
If we can find such a counterexample, then the algorithm terminates with a witness
to the violation of the property.  However, even if the optimization procedure fails to find a true counterexample, the result $x_*$ of the optimization problem can still convey useful information. In particular, our method uses the learned verification policy $\policy_\theta$ to map 
all available information, including $x_*$, to a promising abstract domain $\mathcal{A}$ to use when attempting to verify the property. If the property can be verified using domain $\mathcal{A}$, then the algorithm successfully terminates with a robustness proof.

In cases where the property is neither verified nor refuted,  our algorithm uses the verification policy $\policy_\theta$ to \emph{split} the input region $I$ into two sub-regions $I_1, I_2$ such that $I = I_1 \cup I_2$ and tries to verify/falsify the robustness of each region separately. This form of refinement is useful for both the abstract interpreter as well as the counterexample finder. In particular, since gradient-based optimization methods are not guaranteed to find a global optimum, splitting the input region into smaller parts makes it more likely that the optimizer can find an adversarial counterexample. Splitting the input region is similarly useful for the abstract interpreter because the amount of imprecision introduced by the abstraction is correlated with the size of the input region.

As illustrated by the discussion above, a key part of our verification algorithm
is the use of a  policy $\policy_\theta$ to decide (a) which abstract domain to
use for verification, and (b) how to split the input region into two
sub-regions. Since there is no obvious choice for either the abstract domain or
the splitting strategy, our algorithm takes a \emph{data-driven approach} to
learn a suitable verification policy $\policy_\theta$ during a training phase.
During this training phase,  we use a black-box optimization technique known as Bayesian optimization to
learn values of $\theta$ that lead to strong performance on a representative set of verification problems.  Once this phase is over, the algorithm can be deployed on networks and properties that have not been encountered during training.

Our proposed verification algorithm has some appealing theoretical properties in
that  it is both sound and $\delta$-complete~\cite{delta-complete}. That is, if
our method verifies the property $(I, K)$ for network $\network$, this means
that $\network$ does indeed classify all points in the input region $I$ as
belonging to class $K$. Furthermore, our method is $\delta$-complete in the
sense that, if the property is falsified with counterexample $x_*$, this means
that $x_*$ is within $\delta$ of being a true counterexample.

We have implemented the proposed method in a tool called \toolname\ \footnote{\underline{C}omplete \underline{H}ybrid \underline{A}bstraction \underline{R}efinement and
\underline{O}ptimization  for \underline{N}eural Networks.}, and used it to
analyze hundreds of robustness properties of ReLU neural networks, including both
fully-connected and convolutional neural networks, trained on the MNIST~\cite{LeNet} and CIFAR~\cite{CIFAR} datasets. We have also compared our method against state-of-the-art network verification tools (namely, \reluplex, \reluval, and \ai) and shown that our method outperforms all prior verification techniques, either in terms of accuracy or performance or both.  In addition, our experimental results reveal the benefits of learning to couple proof search and optimization.


In all, this paper makes the following key contributions:
\begin{itemize}[leftmargin=*]
  \item We present a new sound and $\delta$-complete  decision procedure that  combines abstract interpretation and 
    gradient-based counterexample search to
prove robustness of deep neural networks.
\item We describe a method for automatically learning  verification policies that direct counterexample search and abstract interpretation steps carried out during the analysis. 
\item We conduct an extensive experimental evaluation on hundreds of benchmarks and show that our method significantly outperforms state-of-the-art tools 
for verifying neural networks. For example, our method solves $2.6\times$ and $16.6\times$ more benchmarks compared to \reluval\ and \reluplex\ respectively.
\end{itemize}

\paragraph{Organization.} The rest of this paper is organized as follows. In
Section~\ref{sec:background}, we provide necessary background on neural
networks, robustness, and abstract interpretation of neural networks. In
Section~\ref{sec:refinement}, we present our algorithm for checking robustness,
given a verification policy (i.e., the deployment phase).
Section~\ref{sec:learning} describes our data-driven approach for learning a
useful verification policy from training data (i.e., the training phase), and
Section~\ref{sec:theorems} discusses the theoretical properties of our
algorithm. 
Finally, Sections~\ref{sec:impl} and ~\ref{sec:eval} describe our implementation and experimental evaluation, Section~\ref{sec:related} discusses related work, and Section~\ref{sec:conclusion} concludes.

\section{Background}
\label{sec:background}

In this section, we provide some background on neural networks and robustness. 

\subsection{Neural Networks}

A neural network is a 
function $\network:\real^n \to
\real^m$ of the form $\layer_1 \circ \sigma_1 \circ \dots \circ \sigma_{k -1}
\circ \layer_k$, where each $\layer_i$ is a differentiable {\em layer} 
and
each $\sigma_i$ 
is a non-linear, almost-everywhere differentiable \emph{activation
function}.  While there are many types of activation functions, the most popular choice in modern neural networks is the \emph{rectified linear unit (ReLU)}, defined as $\relu(x) = \max(x,0)$. This function is applied element-wise to the output of each layer
except the last. In this work, we consider feed-forward and convolutional networks, which
have the additional property of being Lipschitz-continuous~\footnote{Recall that a
  function is  Lipschitz-continuous if there exists a positive real
  constant $M$ such that, for all  $x_1, x_2$, we have $|f(x_1) -
  f(x_2)| \leq M |x_1 - x_2|$}.

For the purposes of this work, we think of each layer $\layer_i$ as an affine
transformation $(\weights, \bias)$ where $\weights$ is a \emph{weight matrix}
and $\bias$ is a \emph{bias vector}.  Thus, the output of the $i$'th layer is
computed as $\vec{y} = \weights \vec{x} + \bias$. We note that both
\emph{fully-connected} as well as \emph{convolutional layers} can be expressed
as affine transformations~\cite{ai2}. While our approach can also handle other types of layers (e.g., max pooling), we only focus on affine transformations to simplify presentation.

In this work, we consider networks used for classification tasks. That is, given some
input $x \in \real^{n}$, we wish to put $x$ into one of $m$
\emph{classes}. The output of the network $y \in \real^m$ is interpreted as a
vector of \emph{scores}, one for each class. Then $x$ is put into the class with
the highest score. More formally, given some input $x$, we say the network
$\network$ assigns $x$ to a class $K$ if $(\network(x))_K > (\network(x))_j$ for
all $j \neq K$.

\begin{figure}
\includegraphics[scale=0.35]{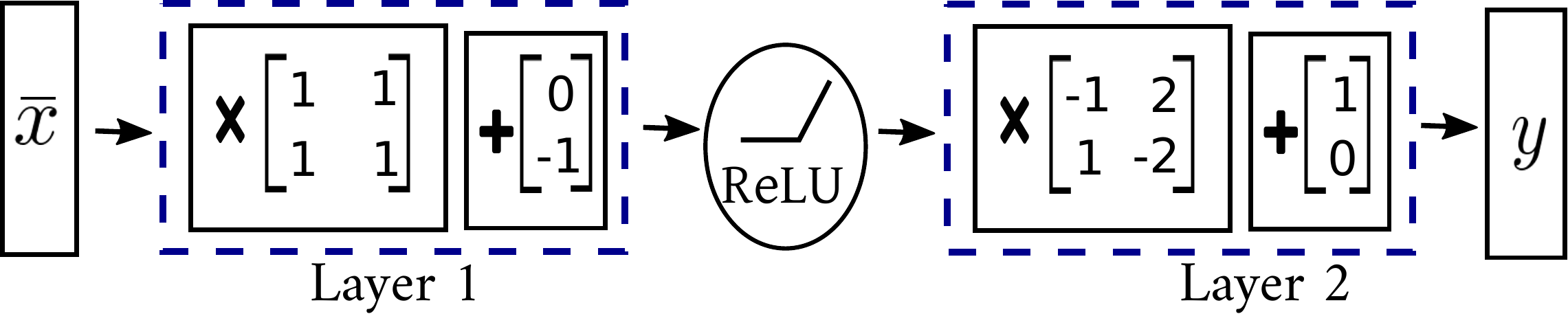}
\vspace{-0.1in}
\caption{A feedforward network implementing XOR}\label{fig:xor}
\vspace{-0.1in}
\end{figure}

\begin{example} 
Figure~\ref{fig:xor} shows a 2-layer feedforward neural network implementing the XOR function. To see why this network
``implements'' XOR, consider the vector $[0 \ 0]^\transpose$. After applying the
affine transformation from the first layer, we obtain $[0 \
-1]^\transpose$. After applying ReLU, we get $[0 \ 0]^\transpose$. Finally,
after applying the affine transform in the second layer, we get
$[1 \ 0]^\transpose$. Because the output at index zero is greater than the
output at index one, the network will classify $[0 \ 0]^\transpose$ as a zero.
Similarly, this network classifies  both  $[0 \
1]^\transpose$ and $[1 \ 0]^\transpose$ as $1$ and  $[1 \
1]^\transpose$ as $0$.
\end{example}

\begin{figure*}[t]
  \includegraphics[width=\textwidth]{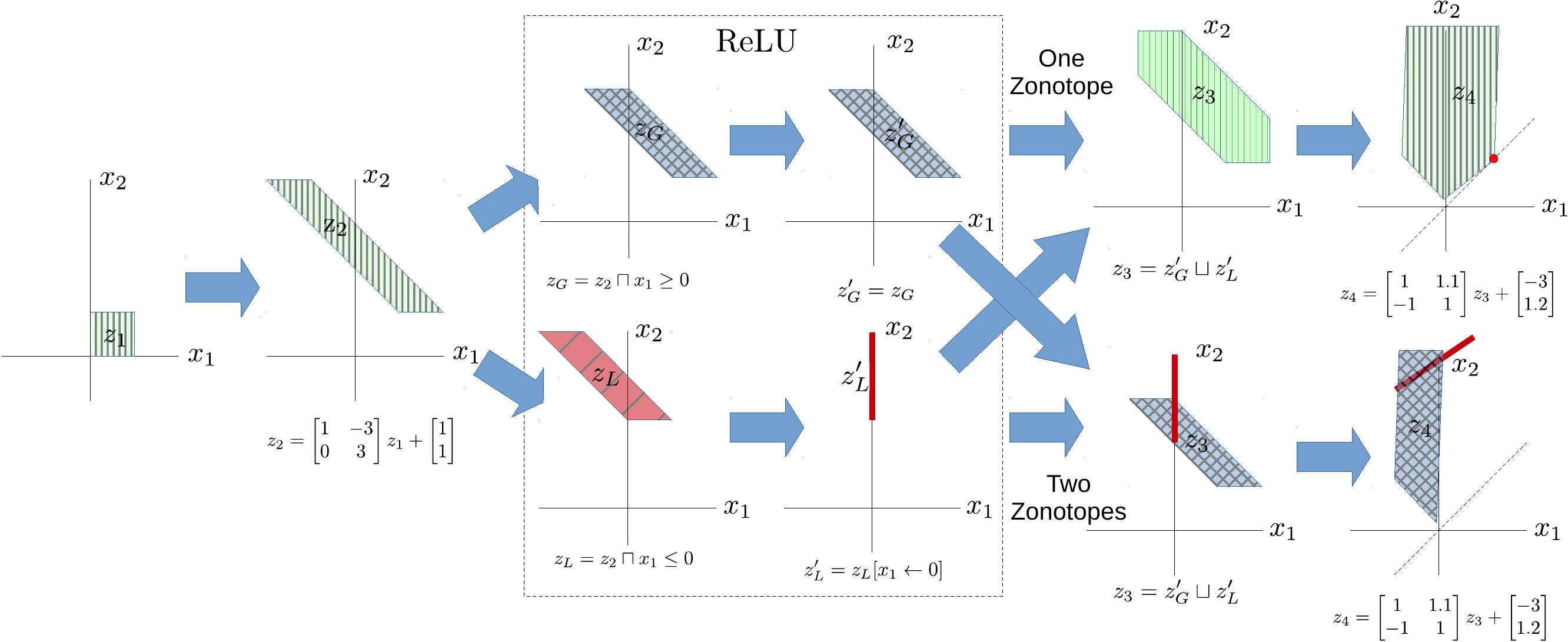}
  \vspace{-0.4in}
  \caption{Zonotope analysis of a neural network.}
  \label{fig:zonotope}
\end{figure*}

\subsection{Robustness}

(Local) robustness~\cite{BastaniILVNC16} is a key correctness property of neural
networks which requires that all inputs within some
region of the input space fall within the same region of the output space.
Since we focus on networks designed for classification tasks,  we
will define ``the same region of the output space'' to mean the region which
assigns the same class to the input. That is, a robustness property asserts that
a small change in the input cannot change the class assigned to that input.

More formally, a \emph{robustness property} is a pair $(I,K)$ with $I \subseteq
\real^{n}$ and $0 \le K \le m-1$. Here, $I$ defines some region of the input that
we are interested in and $K$ is the class into which all the inputs in $I$
should be placed. A network $\network$ is said to satisfy a robustness
property $(I,K)$ if for all $x \in I$, we have $(\network(x))_K >
(\network(x))_j$ for all $j \neq K$.

\begin{example}
  \label{ex:robust}
Consider the following network with two  layers, i.e.,
$\network(x) = W_2(\relu(W_1 x + b_1)) + b_2$ where:
\[W_1 = \begin{bmatrix} 1 \\ 2 \end{bmatrix} \quad b_1 = \begin{bmatrix} -1 \\
  1 \end{bmatrix} \quad W_2 = \begin{bmatrix} 2 & 1 \\ -1 & 1 \end{bmatrix}
  \quad b_2 = \begin{bmatrix} 1 \\ 2 \end{bmatrix}\]
For the input $x=0$, we have $\network(0) = [ 1 \ 3 ]^\transpose$; thus, the
network outputs label 1 for input $0$. Let $I = [-1, 1]$ and $K = 1$. Then for all $x \in I$, 
the output of $\network$ is of the form $[a+1\ a+2]^T$
for some $a \in [0,3]$.
Therefore, the network classifies every point in $I$ as belonging to class 1, meaning that the network is robust in $[-1,1]$. On the other hand,
suppose we extend this interval to $I' = [-1,2]$. Then $\network(2) = [ 8
  \ 6 ]^\transpose$, so $\network$ assigns input  $2$ as belonging to class 0. Therefore $\network$ is \emph{not} robust in the input region
$[-1,2]$.
\end{example}

\subsection{Abstract Interpretation for Neural Networks}\label{sec:ai2}

In this paper, we build on the prior \ai\ work~\cite{ai2} for analyzing neural networks 
using the framework of \emph{abstract interpretation}~\cite{cousot77}. \ai\ allows analyzing neural networks using a variety of numeric abstract domains, including intervals (boxes)~\cite{cousot77}, polyhedra~\cite{polyhedra}, and zonotopes~\cite{zonotope}. In addition, \ai\ also supports \emph{bounded powerset domains} ~\cite{cousot77}, which essentially allow a bounded number of disjunctions in the abstraction. Since the user can specify any number of disjunctions, there are \emph{many} different abstract domains to choose from, and the precision and scalability of the analysis crucially depend on one's choice of the abstract domain.



The following example illustrates a robustness property that can be verified using the bounded zonotope domain with two disjuncts but not with intervals or plain zonotopes:

\begin{example}
  \label{ex:zonotope}
  Consider a network defined as:
  \[\network(x) = \begin{bmatrix} 1 & 1.1 \\ -1 & 1 \end{bmatrix}
  \relu\left(\begin{bmatrix} 1 & -3 \\ 0 & 3 \end{bmatrix} x + \begin{bmatrix} 1
  \\ 1 \end{bmatrix}\right) + \begin{bmatrix} -3 \\ 1.2 \end{bmatrix}\]
  As in the previous example,  the first index in the
  output vector corresponds to class $A$ and the second index corresponds to 
  class $B$. Now, suppose we want to verify that for all $x \in
  [0,1]^2$, the network assigns class $B$ to $x$. 
  
  %
  Let us now analyze this network using the
  zonotope abstract domain, which overapproximates a given region using a zonotope (i.e., center-symmetric polytope).
  The analysis of this network using the zonotope domain is illustrated in Figure~\ref{fig:zonotope}.
  At first, the initial region is propagated through the affine transformation
  as a single zonotope. Then, this zonotope is split into two pieces, the blue
  (crosshatched)
  one for which $x_1 \ge 0$ and the red (diagonally striped) one for which $x_1 \le 0$. The
  ReLU transforms the red piece into a line. (We omit the ReLU over $x_2$
  because it does not change the zonotopes in this case.) After the ReLU, we
  show two cases: on top is the plain zonotope domain, and on the bottom is a
  powerset of zonotopes domain. In the plain zonotope domain, the abstraction after the ReLU is
  the join of the blue and red zonotopes, while in the powerset
  domain we  keep the blue and red zonotopes separate. The final images show
  how the second affine transformation affects all three  zonotopes.
  
This example illustrates that the propery cannot be verified using the plain
zonotope domain, but it \emph{can} be verified using the powerset of zonotopes
domain.   Specifically, observe that the green (vertically striped) zonotope at the top includes the point
	$[1.2 \ 1.2]^\transpose$ (marked by a dot), where the robustness specification is violated.
On the other hand, the blue and red zonotopes obtained using the powerset domain do not contain any unsafe points,
  so the property is verified using this more precise abstraction.
\end{example}

\section{Algorithm for  Checking Robustness}
\label{sec:refinement}


In this section, we describe our algorithm for checking robustness properties of
neural networks. Our algorithm  interleaves
optimization-based counterexample search with proof synthesis using abstraction refinement.
At a high level, abstract
interpretation provides an efficient way to verify properties but is subject to
false positives. Conversely, optimization based techniques for finding counterexamples
 are efficient for finding adversarial inputs, but suffer from false
negatives. Our algorithm combines the strengths of these two techniques by searching for both proofs and counterexamples at the same time
and using information from the counterexample search to guide proof
search.

Before we describe our algorithm in detail, we need to define our optimization
problem more formally. Given a network $\network$ and a robustness property $(I,
K)$, we can view the search for an
adversarial counterexample as the following optimization
problem:
\begin{equation}\label{eq:optimization}
x_* = \argmin_{x \in I} \left( \objective(x)\right)
\end{equation}
where our \emph{objective function} $\objective$ is defined as follows:
\begin{equation}\label{eq:objective}
  \mathcal{F}(x) = (\network(x))_K - \max_{j \neq K} (\network(x))_j
\end{equation}

Intuitively, the objective function $\objective$ measures the difference between the score for class $K$ and the maximum score among classes 
other than $K$. Note that if the value of this objective function is not
positive at some point $x$, then there exists some class which has a greater (or equal) score
than the target class, so  point $x$ constitutes a true adversarial counterexample.

The optimization problem from Eq.~\ref{eq:optimization} is clearly useful for searching for counterexamples to the robustness property. However, even if the solution $x_*$  is  not a true counterexample (i.e., $\objective(x_*) >0$), we can still use the result of the optimization problem to guide proof search.

Based on this intuition, we now explain our decision procedure, shown in 
Algorithm~\ref{alg:refinement}, in more detail. The {\sc Verify} procedure takes as input a
network $\network$, a robustness property $(I, K)$ to be verified, and
a so-called \emph{verification policy} $\policy_\theta$. 
As mentioned in Section~\ref{sec:intro}, 
the verification policy is used to decide what kind of abstraction to
use and how to split the input region when attempting to verify the
property. In more detail, the verification policy $\policy_\theta$,
parameterized by $\theta$, is a pair 
$(\apolicy_{\theta}, \ipolicy_{\theta})$, where $\apolicy_{\theta}$ is
a (parameterized) function 
known as the {\em domain policy} and $\ipolicy_{\theta}$ is a function
known as
the {\em partition policy}. 
The domain policy is 
used to decide which abstract domain to use, while 
the partition policy determines how to split the input region $I$ into two partitions to be analyzed separately. 
In general, it is quite difficult to write a good verification policy by hand
because there are many different parameters to tune and neural networks are
quite opaque and difficult to interpret.
In Section~\ref{sec:learning}, we explain how the parameters of these policy functions are learned from data.

\begin{algorithm}[t]
  \begin{algorithmic}[1]
    \Procedure{Verify}{$\network, I, K, \policy_\theta$}
    \vspace{0.05in}
    \Statex \Input{A network $N$, robustness property $(I, K)$
     and verification policy $\policy_\theta = (\apolicy_\theta, \ipolicy_\theta)$}
    \Statex \Output{Counterexample if $\network$ is not robust,
    or {\rm Verified}.}
    \vspace{0.05in}
    \State $x_* \gets \Call{Minimize}{I, \objective}$
    \If{$\objective(x_*) \le 0$}
    \State \textbf{return} $x_*$
    \EndIf
    \State $\mathcal{A} \gets \apolicy_\theta(\network,I,K,x_*)$
    \If{$\Call{Analyze}{\network, I, K, \mathcal{A}} = {\rm Verified}$}
    \State \textbf{return} {\rm Verified}
    \EndIf
    \State $(I_1, I_2) \gets \ipolicy_\theta(\network, I, K, x_*)$
    \State $r_1 \gets \Call{Verify}{\network, I_1, K, \policy_\theta}$
    \If{$r_1 \neq \text{Verified}$}
    \State \textbf{return} $r_1$
    \EndIf
    \State \textbf{return} $\Call{Verify}{\network, I_2, K,
    \policy_\theta}$
    \EndProcedure
  \end{algorithmic}
  \caption{The main algorithm}
  \label{alg:refinement}
\end{algorithm}

At a high-level, the {\sc Verify} procedure works as follows: First, we  try to find a counterexample to the given robustness property by solving the optimization problem from Eq.~\ref{eq:optimization} using the well-known \emph{projected gradient descent (PGD)} technique. If $\objective(x_*)$ is non-positive, we have found a true counterexample, so the algorithm produces $x_*$ as a witness to the violation of the property.
Otherwise, we try to verify the property 
using abstract interpretation. 

As mentioned in Section~\ref{sec:background}, there are many different abstract domains that can be used to verify the property, and the choice of the abstract domain has a huge impact on the success and efficiency of verification. Thus, our approach leverages the domain policy $\apolicy_\theta$ to choose a sensible abstract domain to use when attempting to verify the property. Specifically, the domain policy $\apolicy_\theta$ takes as input the network $\network$, the robustness specification $(I, K)$, and the solution $x_*$ to the optimization problem and chooses an abstract domain $\mathcal{A}$ that should be used for attempting to prove the property. If the property can be verified using domain $\mathcal{A}$, the algorithm terminates with a proof of robustness.

In cases where the property is neither verified nor refuted in the current iteration, the algorithm makes progress by splitting the input region $I$ into two disjoint partitions $I_1, I_2$ such that $I = I_1 \uplus I_2$. The intuition is that, even if we cannot prove robustness for the whole input region $I$, we may be able to increase analysis precision by performing a case split. That is, as long as all points in \emph{both} $I_1$ and $I_2$ are classified as having label $K$, this means that all points in $I$ are also assigned label $K$ since we have $I = I_1 \cup I_2$. In cases where the property is false, splitting the input region into two partition can similarly help adversarial counterexample search because gradient-based optimization methods do not always converge to a global optimum.

Based on the above discussion, the key question is how to partition the
input region $I$ into two regions $I_1, I_2$ so that each of $I_1, I_2$ has a
good chance of being verified or falsified. Since this question again does not have an obvious answer, we utilize our \emph{partition policy} $\ipolicy_\theta$ to make this decision. Similar to the domain policy, $\ipolicy_\theta$ takes as input the network, the property, and the solution $x_*$ to the optimization problem and ``cuts'' $I$ into two sub-regions $I_1$ and $I_2$ using a hyper-plane. Then, the property is verified if and only if the recursive call to {\sc Verify} is succsessful on both regions.

\begin{figure}
  \includegraphics[scale=0.25]{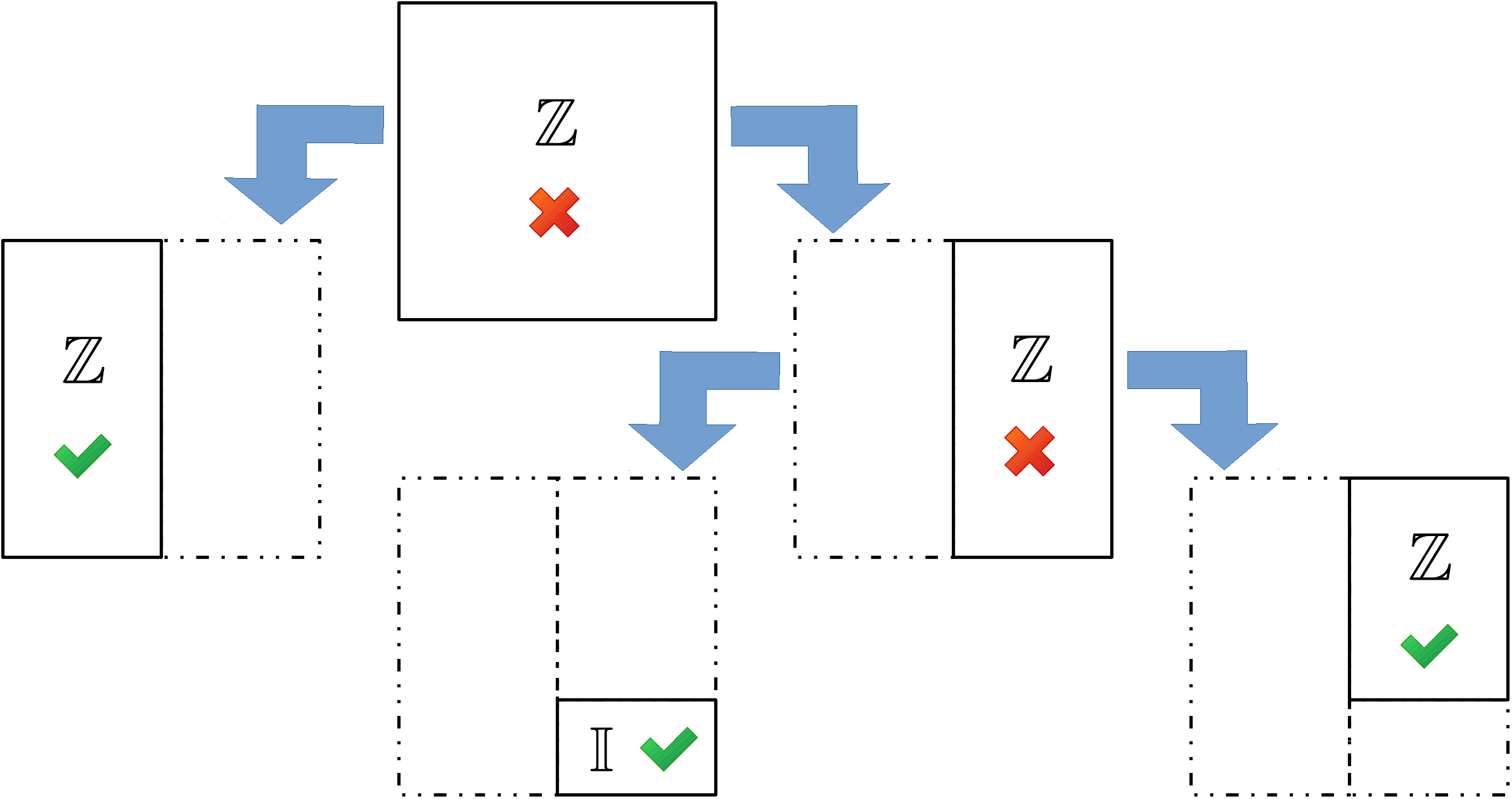}
  \vspace{-0.1in}
  \caption{The splits chosen for Example~\ref{ex:split}.}
  \label{fig:splitting}
  \vspace{-0.15in}
\end{figure}

\begin{example}
  \label{ex:split}
  Consider the XOR network from Figure~\ref{fig:xor} and the robustness property $([0.3,0.7]^2, 1)$. That is, for all inputs $\vec{x}$ with $0.3 \le x_1,x_2 \le 0.7$,  $\vec{x}$ should be assigned to class 1 (assume classes are zero-indexed).
We  now illustrate how Algorithm~\ref{alg:refinement} verifies this property using the plain interval 
and zonotope abstract domains. The process is illustrated in
Figure~\ref{fig:splitting}, which shows the splits made in each iteration  as well as the domain used
to analyze each region ($\mathbb{Z}$ denotes zonotopes, and $\mathbb{I}$ stands for
intervals).

 Algorithm~\ref{alg:refinement} starts by searching for an adversarial counterexample, but fails to find one since the property actually holds.
 Now, suppose that our domain policy $\apolicy_\theta$  chooses zonotopes to try to verify the property.
 Since the property cannot be verified using zonotopes, the call to \Call{Analyze}{} will fail. Thus, we now consult
 the partition policy $\ipolicy_\theta$ to split  this region into two
  pieces $I_1 = [0.3, 0.5] \times [0.3, 0.7]$ and
  $I_2 = [0.5, 0.7] \times [0.3, 0.7]$.

  Next, we recursively invoke Algorithm~\ref{alg:refinement} on both sub-regions
  $I_1$ and $I_2$. Again, there is no counterexample for either region, so we use the domain policy to choose 
  an abstract domain for each of  $I_1$ and $I_2$. Suppose that $\apolicy_\theta$ yields the zonotope domain for both $I_1$ and $I_2$. Using this domain, we can verify robustness in $I_1$ but not in $I_2$. Thus, for the second sub-problem, we again consult $\ipolicy_\theta$ to obtain two sub-regions $I_{2,1} = [0.5, 0.7] \times [0.3,
  0.42]$ and $I_{2,2} = [0.5, 0.7] \times [0.42, 0.7]$ and determine using
  $\apolicy_\theta$ that $I_{2,1}, I_{2,2} $ should be analyzed using intervals
  and zonotopes respectively. Since robustness can be verified using these
  domains, the algorithm successfully terminates. Notice that  the three verified subregions cover the entire initial
  region.
\end{example}

\section{Learning a Verification Policy}
\label{sec:learning}


As described in Section~\ref{sec:refinement}, our decision procedure for checking robustness uses 
a verification policy $\policy_\theta = (\apolicy_\theta,
\ipolicy_\theta)$  to  choose a suitable abstract domain and an input
partitioning strategy. In this section, we discuss  our policy
representation and  how to learn values of  $\theta$ that lead to
good performance.

\subsection{Policy Representation}\label{sec:policy-representation}

In this work, we implement verification  policies $\apolicy_\theta$ and $\ipolicy_\theta$ using a function of the following shape:
\begin{equation}\label{eq:policy-rep}
  \varphi(\theta \, \rho(\network, I, K, x_* ))
\end{equation}
where $\rho$ is a \emph{featurization function} that extracts a feature vector from the input,  $\varphi$ is a \emph{selection} function that converts a
  real-valued vector to a suitable output (i.e., an abstract domain
 for $\apolicy_\theta$ and the two subregions for $\ipolicy_\theta$), and $\theta$ corresponds to a parameter matrix that is automatically learned from a representative set of training data.
We discuss our featurization and selection functions in this sub-section and explain how to learn parameters $\theta$ in the next sub-section.

\paragraph{Featurization.} \label{sec:featurization} 
%
%
As standard in machine learning, we need  to convert the input $\iota = (\network, I, K, x_*)$ to a feature
vector.  
Our choice of features is influenced by our insights about the
verification problem, and we deliberately use a small number of features
for two reasons:  First, a large number of dimensions can lead to overfitting and
poor generalization (which is especially an issue when
training data is fairly small). Second,
a high-dimensional feature vector leads to a more difficult
learning problem, and contemporary Bayesian optimization engines
 only scale to a few tens of dimensions.  

Concretely, 
our featurization function considers several kinds of information, including: 
(a) the behavior of the network near $x_*$, (b) where $x_*$ falls
in the input space, and (c) the size of the input space. Intuitively, we expect
that (a) is useful because as $x_*$ comes closer to violating the specification,
we should need a more 
precise abstraction, while (b) and (c) inform how we should
split the input region during refinement. Since the precision of the
analysis is correlated with how the split is performed, we found the
same featurization function to work well for both policies
$\apolicy_\theta$ and $\ipolicy_\theta$.  In Section~\ref{sec:impl},
we discuss the exact features used in our implementation.

\paragraph{Selection function.}  Recall that the purpose of the selection function $\varphi$ is to convert $\theta \rho(\iota)$ to a "strategy", which is an abstract domain for $\apolicy$ and a hyper-plane  for $\ipolicy$.  Since the strategies for these two functions are quite different, we use two different selection functions, denoted $\varphi_\alpha, \varphi_I$ for the domain and partition policies respectively. 

The selection function $\varphi_\alpha$ is quite simple and maps $\theta
\rho(\iota)$ to a tuple $(d, k)$ where $d$ denotes the base abstract domain
(either intervals $\mathbb{I}$ or zonotopes $\mathbb{Z}$ in our implementation)
and $k$ denotes the number of disjuncts. Thus, $(\mathbb{Z}, 2)$ denotes the
powerset of zonotopes abstract domain, where the maximum number of disjuncts is
restricted to $2$, and $(\mathbb{I}, 1)$ corresponds to the standard interval
domain.

In the case of the partition policy $\ipolicy$, the  selection function $\varphi_I$  is also a
tuple $(d, c)$ where $d$ is the dimension along which we split the input region
and $c$ is the point at which  to split. In other words, if $\varphi_I(\theta \rho(\iota)) = (d, c)$, this means that we split the input region $I$ using the hyperplane $x_d = c$.  Our selection function $\varphi_I$ does not consider arbitrary hyperplanes of the form $c_1 x_1 + \ldots c_n x_n = c$ because splitting the input region along an arbitrary hyperplane may result in sub-regions that are not expressible in the chosen abstract domain. In particular, this is true for both the interval and zonotope domains used in our implementation.

\subsection{Learning using Bayesian Optimization}\label{sec:offline}

As made evident by Eq.~\ref{eq:policy-rep}, the parameter matrix $\theta$ has a huge impact on the choices made by our verification algorithm. However, manually coming up with these parameters is very difficult because 
the right choice of coefficients depends on both the property, the
network, and the underlying abstract interpretation engine. In this
work, we take a data-driven approach to solve this problem and use
{\em Bayesian optimization} to learn a parameter matrix
$\theta$ that leads to optimal performance by the verifier on a set of
training problems. 

\paragraph{Background on Bayesian optimization.}
Given a function $F: \real^n \to \real$, the goal of Bayesian
optimization is to find a vector $\vec{x^*} \in \real^n$ that
maximizes $F$. Importantly, Bayesian optimization does not assume that $F$ is
differentiable; also, in practice, it can achieve reasonable performance without having
to evaluate $F$ very many times. In our setting, the function $F$ represents the
performance of a verification policy. This function is not necessarily
differentiable in the parameters of the verification policy, as a
small perturbation to the policy parameters can lead to the choice of
a different domain. Also, evaluating the function requires
an expensive
round of abstract interpretation. For these reasons, Bayesian
optimization is a good fit to our learning problem. 

At a high level, Bayesian optimization repeatedly samples inputs until a time
limit is reached and returns the best input found so far. However, rather than
sampling inputs {at random}, the key part of Bayesian optimization is to predict
what input is useful to sample next. Towards this goal, the algorithm uses (1) a
\emph{surrogate model} $\mathcal{M}$ that expresses our current belief about
$F$, and (b) an \emph{acquisition function} $\mathcal{A}$ that employs
$\mathcal{M}$ to decide the most promising input to sample in the next
iteration. The surrogate model $\mathcal{M}$ is initialized to capture prior
beliefs about $F$ and is updated based on observations on the sampled points.
The acquisition function $\mathcal{A}$  is chosen to trade off exploration and
exploitation where "exploration" involves sampling points with high
uncertainty, and "exploitation" involves sampling points where $\mathcal{M}$ predicts a high value of $F$. Given model $\mathcal{M}$ and function  $\mathcal{A}$, Bayesian optimization  samples the most promising input $\vec{x}$ according to $\mathcal{A}$, evaluates $F$ at $\vec{x}$, and updates the statistical model $\mathcal{M}$ based on the observation $F(\vec{x})$. This process is repeated until a time limit is reached, and the best  input  sampled so far is returned as the optimum. We refer the reader to ~\cite{bayesian-opt}  for a more detailed overview of Bayesian optimization.

\paragraph{Using Bayesian optimization.} In order to apply Bayesian optimization
to our setting, we first need to define what function we want to optimize.
Intuitively, our objective function should estimate the quality of the analysis
results based on  decisions made by verification policy $\policy_\theta$.
Towards this goal, we fix a set $S$ of representative training problems that can
be used to estimate the quality of $\policy_\theta$. Then, given a parameters
matrix $\theta$, our objective function $F$ calculates a score based on (a) how
many benchmarks in $S$ can be successfully solved within a given time limit, and
(b) how long it takes to solve the benchmarks in $S$. More specifically, our
objective function $F$ is parameterized by a time limit $t \in \real$ and penalty
$p \in \real$ and calculates the score for a matrix $\theta$ as follows:
\[
F(\theta) = -\sum_{s \in S} \emph{cost}_\theta(s)
\vspace{-0.1in}
\]
where:
\[
\emph{cost}_\theta(s) = \left \{
\begin{array}{ll}
\emph{Time}(\emph{Verify}_\theta(s)) & {\rm if} \ s \ {\rm solved} \ {\rm within}  \ t \\
p \cdot t & {\rm otherwise}
\end{array}
\right .
\]

Intuitively,  $p$ controls how much we want to penalize failed
verification attempts -- i.e., the higher the value of $p$, the more biased the
learning algorithm  is towards more precise (but potentially slow) strategies.
On the other hand, small values of $p$ bias learning towards  strategies that
yield fast results on the solved benchmarks, even if some of the benchmarks cannot be solved within the given time limit.~\footnote{In our implementation, we choose $p=2$, $t=700s$.}

In order to apply Bayesian optimization to our problem, we also need to choose a suitable  acquisition function and surrogate mode. Following standard practice, we adopt a \emph{Gaussian process}
\cite{gaussian-process} as
our surrogate model and use \emph{expected improvement} \cite{expected-improvement} for the acquisition function.

\section{Termination and Delta Completeness}
\label{sec:theorems}

In this section, we discuss some theoretical  properties of our verification algorithm, including soundness, termination, and completeness. To start with, it is easy to see that Algorithm~\ref{alg:refinement} is sound, as it only returns "Verified" once it establishes that \emph{every} point in the input space is classified as $K$. This is the case because every time we split the input region $I$ into two sub-regions $I_1, I_2$, we ensure that $I = I_1 \cup I_2$, and the underlying abstract interpreter is assumed to be sound. However, it is less clear whether Algorithm~\ref{alg:refinement} always terminates or whether it has any completeness guarantees.

Our first observation is that the {\sc Verify} procedure,
\emph{exactly} as presented in Algorithm~\ref{alg:refinement}, does
{not} have termination guarantees under realistic assumptions about
the optimization procedure used for finding adversarial
counterexamples. Specifically, if the procedure {\sc Minimize} invoked
at line 2 of Algorithm~\ref{alg:refinement} returned a \emph{global}
minimum, then we could indeed guarantee
termination.~\footnote{However, if we  make this assumption, the
  optimization procedure itself would be a sound and complete decision
  procedure for verifying robustness!} However,
since gradient-based optimization procedures do not have this property, Algorithm~\ref{alg:refinement} may not be able to find a true adversarial counterexample even as we make the input region infinitesimally small. Fortunately, we can guarantee termination and a form of completeness (known as $\delta$-completeness) by making a very small change to Algorithm~\ref{alg:refinement}.

To guarantee termination of our verification algorithm, we will make the following slight change to line 3 of Algorithm~\ref{alg:refinement}: Rather than checking  $
  \objective(x_*) \le 0
$
(for $\objective$ as defined in Eq.~\ref{eq:objective}) we will instead check: 
\begin{equation}\label{eq:modification}
\objective(x_*) \le \delta
\end{equation}
While this modification can cause our verification algorithm to produce false positives under certain  pathological conditions, the analysis can be made as precise as necessary by picking a value of $\delta$ that is arbitrarily close to $0$.  Furthermore, under this change, we can now prove  termination under some mild and realistic assumptions. In order to formally state these assumptions, we first introduce
the following notion of the \emph{diameter} of a region:
\begin{definition}
  For any set $X \subseteq \real^n$, its \emph{diameter} $D(X)$ is defined as $\sup\{\|x_1 - x_2\|_2 \mid x_1,x_2 \in X\}$
  if this value exists. Otherwise the set is said to have infinite diameter.
  \label{def:diameter}
\end{definition}

We now use this notion of diameter to state two key assumptions that are needed to prove termination:

\begin{assumption}
  There exists some $\lambda \in (0,1)$   such that for any network $\network$, input region $I$, and point $x_* \in
  I$, if $ \ipolicy(N, I, x_*) = (I_1, I_2) $, then $D(I_1) < \lambda
  D(I)$ and $D(I_2) < \lambda D(I)$.
  \label{asm:split-reduces-size}
\end{assumption}

Intuitively, this assumption states that
the two resulting subregions after splitting are smaller than the original region by some
factor $\lambda$. It is easy to enforce this condition on any partition policy 
by choosing a hyper-plane $x_d = c$ where $c$ is not at the boundary of the input region.



Our second assumption concerns the abstract domain:
\begin{assumption}
  Let $\network^\#$ be the abstract transformer representing a network
  $\network$. 
  For a given input region $I$, 
  we assume there exists some $K_\network \in \real$ such
  that $D(\gamma(\network^\#(\alpha(I)))) < K_\network D(I)$.
  \label{asm:finite-overapprox}
\end{assumption}
This assumption asserts that the Lipschitz continuity of the network extends to
its abstract behavior. Note that this assumption holds in several numerical
domains including intervals, zonotopes, and powersets thereof.

\begin{theorem}
Consider the variant of Algorithm~\ref{alg:refinement} where the predicate at line 3 is replaced with
Eq.~\ref{eq:modification}. Then, if the input region has finite diameter, the verification algorithm always terminates
under Assumptions~\ref{asm:split-reduces-size} and~\ref{asm:finite-overapprox}.
\label{thm:termination}
\ifextended
  \footnote{Proofs for all theorems can be found in the appendix.}
  \else
\footnote{Proofs for all theorems can be found in the
  extended version of the paper on arXiv.}
  \fi
\end{theorem}

In addition to termination, our small modification to Algorithm~\ref{alg:refinement} also ensures a property called $\delta$-completeness~\cite{delta-complete}. In the context of  satisfiability over real numbers, $\delta$-completeness means that, when the algorithm returns a satisfying assignment $\sigma$, the formula is either indeed satisfiable or a $\delta$-perturbation on its numeric terms would make it satisfiable. To adapt this notion of $\delta$-completeness to our context, we introduce the folowing concept $\delta$-counterexamples:

\begin{definition}
For a given network $\network$, input region $I$, target class $K$, and
$\delta>0$, a $\delta$-counterexample is a point $x \in I$ such that for some
$j$ with $1 \le j \le m$ and $j \neq K$, $\network(x)_K - \network(x)_j \le
\delta$.
\label{def:delta-complete}
\end{definition}
Intuitively, a $\delta$-counterexample is a point in the input space 
for which the output almost violates the given specification. We can
view $\delta$ as a parameter which controls how close to violating the
specification a point must be to be considered ``almost'' a counterexample.

\begin{theorem}
Consider the variant of Algorithm~\ref{alg:refinement} where the predicate at line 3 is replaced with
Eq.~\ref{eq:modification}. Then, the verification algorithm is
$\delta$-complete --- i.e., if the property is not verified,
it returns a {$\delta$-counterexample}.
\label{thm:completeness}
\end{theorem}

\section{Implementation}
\label{sec:impl}

We have implemented the ideas proposed in this paper in a tool called \toolname, written
in C++. Internally, \toolname\ uses the ELINA abstract interpretation library \cite{elina} to
implement the \Call{Analyze}{} procedure from Algorithm~\ref{alg:refinement}, and it uses the BayesOpt library
~\cite{bayesopt} to perform Bayesian optimization.


\paragraph{Parallelization.}
Our proposed verification algorithm is easily parallelizable, as different 
calls to the abstract interpreter can be run on different threads. Our implementation takes
advantage of this observation and utilizes as many threads as the host  machine can provide by running 
different calls to ELINA in parallel.


\paragraph{Training.}
We trained our verification policy on 12 different robustness properties of
a neural network used in the ACAS Xu collision avoidance system~\cite{uas}.
However, since even verifying even a single benchmark can take a very long time,
our implementation uses two tactics to reduce training time. First, we parallelize the 
 training phase of the algorithm using the MPI framework  \cite{MPI} and solve each benchmark
at the same time.  Second, we set a time limit of 700 seconds (per-process cputime) per benchmark. 
Contrary to what we may expect from machine learning systems, a small set of
benchmarks is sufficient to learn a good strategy for our setting. We conjecture
that this is because the relatively small number of features allowed by Bayesian
optimization helps to regularize the learned policy.


\paragraph{Featurization.} Recall that our verification policy uses a \emph{featurization function} to convert its
input to a feature vector. As mentioned in Section~\ref{sec:featurization}, this
featurization function should select a compact set of features so that our
training is efficient and avoids overfitting our policy to the training set.
These features should also capture revalant information about the network and
the property so that our learned policy can generalize across networks. With
this in mind, we used the following features in our implementation:

\begin{itemize}[leftmargin=*]
  \item the distance between the center of the input region $I$ and
    the solution $x_*$ to the optimization problem
  \item the value of the objective function $\objective$ (Eq.~\ref{eq:objective}) at $x_*$ 
  \item the magnitude of the gradient of the network at $x_*$
  \item average length of the input space along each dimension
\end{itemize}

\paragraph{Selection.} Recall from Section~\ref{sec:learning} that our verification policy $\policy$ uses two different selection functions $\varphi^\alpha$ and $\varphi^I$ for choosing an abstract domain and splitting plane respectively.

The selection function $\varphi^I$  takes a vector of three inputs. The
first two are real-valued numbers that decide which dimension to split on. Rather
than considering all possible dimensions, our implementation chooses between two
dimensions to make training more manageable. The first one is the longest dimension (i.e.,  input dimension with
the largest length), and the second one is the dimension that has the largest
\emph{influence} \cite{reluval} on $\network(x)_K$. 
The last input to the selection function is
the offset at which to split the region. This value is clipped to
$[0,1]$ and
then interpreted as a ratio of the distance from the center of the input region $I$ to the
solution $x_*$ of Eq.~\ref{eq:optimization}. For example, if the value is 0, the region will be
bisected, and if the value is 1, then the splitting plane will intersect $x_*$.
Finally, if the splitting plane is at the boundary of $I$, it is offset slightly
so that the strategy satisfies Assumption~\ref{asm:split-reduces-size}.

The selection function  $\varphi^\alpha$ for choosing an abstract domain
 takes a vector of two inputs. The first controls the base abstract domain
(intervals or zonotopes) and the second controls the number of disjuncts to use.
In both cases, the output is extracted by first clipping the input to a fixed
range and then discretizing the resulting value.

\section{Evaluation}
\label{sec:eval}

To evaluate the ideas proposed in this paper, we conduct an experimental evaluation that is designed to answer the following three research questions:
\begin{enumerate}[label=(RQ\arabic*)]
  \item How does \toolname\ compare against  state-of-the-art tools for proving neural
    network robustness?
  \item How does counterexample search impact the performance of \toolname?
  \item What is the impact of learning a verification policy on the performance of \toolname?
\end{enumerate}

\paragraph{Benchmarks.} To answer these research questions, we collected a benchmark suite of 602 verification problems
across 7 deep neural networks, including one convolutional network and
several fully connected networks.  The fully connected networks have sizes 3$\times$100, 6$\times$100,
9$\times$100, and
9$\times$200, where $N\times M$ means there are $N$ fully connected
layers and each interior layer has size $M$. The convolutional network has
a LeNet architecture \cite{LeNet} consisting of two convolutional layers, followed by a
max pooling layer,  two more convolutional layers, another max pooling
layer, and finally three fully connected layers.
All of these networks were trained on the
MNIST \cite{LeNet} and CIFAR \cite{CIFAR} datasets.

\subsection{ Comparison with \ai \ (RQ1)}\label{sec:ai2-compare}
\begin{figure*}
  \includegraphics[width=\textwidth]{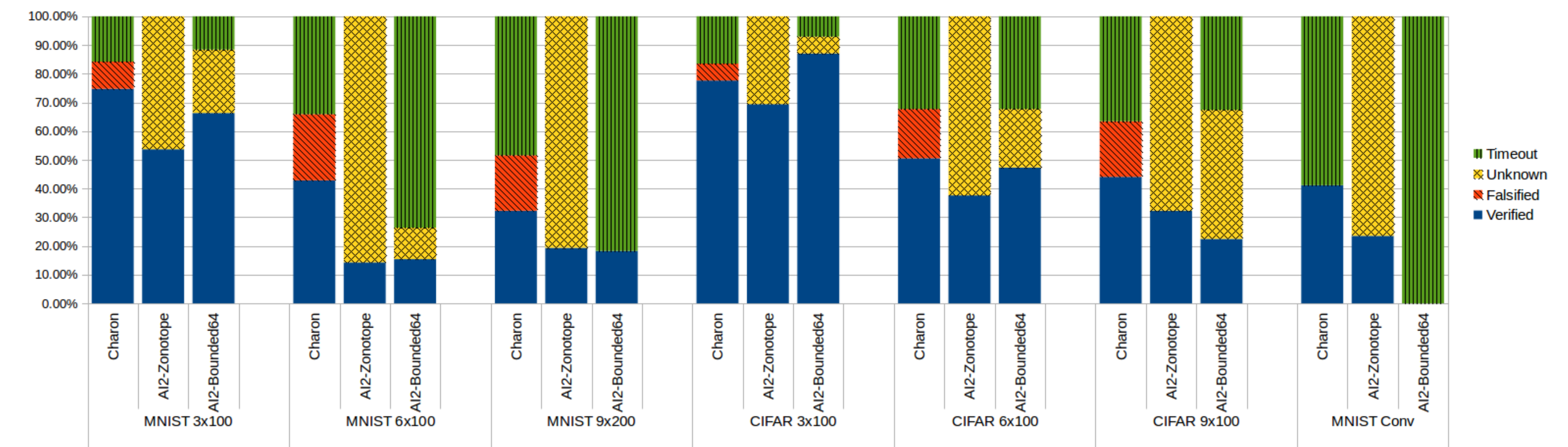}
  \vspace{-0.2in}
  \caption{Summary of results for \ai\ and \toolname.}
  \label{fig:ai2-summary}
  \vspace{-0.2in}
\end{figure*}

For each network, we attempt to
verify around 100 robustness properties.
Following prior work \cite{ai2}, the evaluated robustness properties are
so-called \emph{brightening attacks} \cite{deepxplore}. For an input point $x$ and a threshold
$\tau$, a brightening attack consists of the input region
\[I = \left\{x' \in \real^n \mid \forall i. (x_i \ge \tau \wedge x_i \le x'_i
\le 1) \vee x'_i = x_i\right\}.\]
That is, for each pixel in the input image, if the value of that pixel is
greater than $\tau$, then the corresponding pixel in the perturbed image may be
anywhere between the initial value and one, and all other pixels remain unchanged.

\paragraph{Set-up.} All experiments described in this section were performed on the Google Compute Engine (GCE) \cite{GCP} using an 8 vcpu instance with 10.5 GB of memory. All time measurements report the total CPU time (rather than wall clock time) in order to avoid biasing the results because of \toolname's parallel nature.
For the purposes of this experiment, we set a time limit of 1000 seconds per benchmark.

\begin{figure}
  \includegraphics[width=\columnwidth]{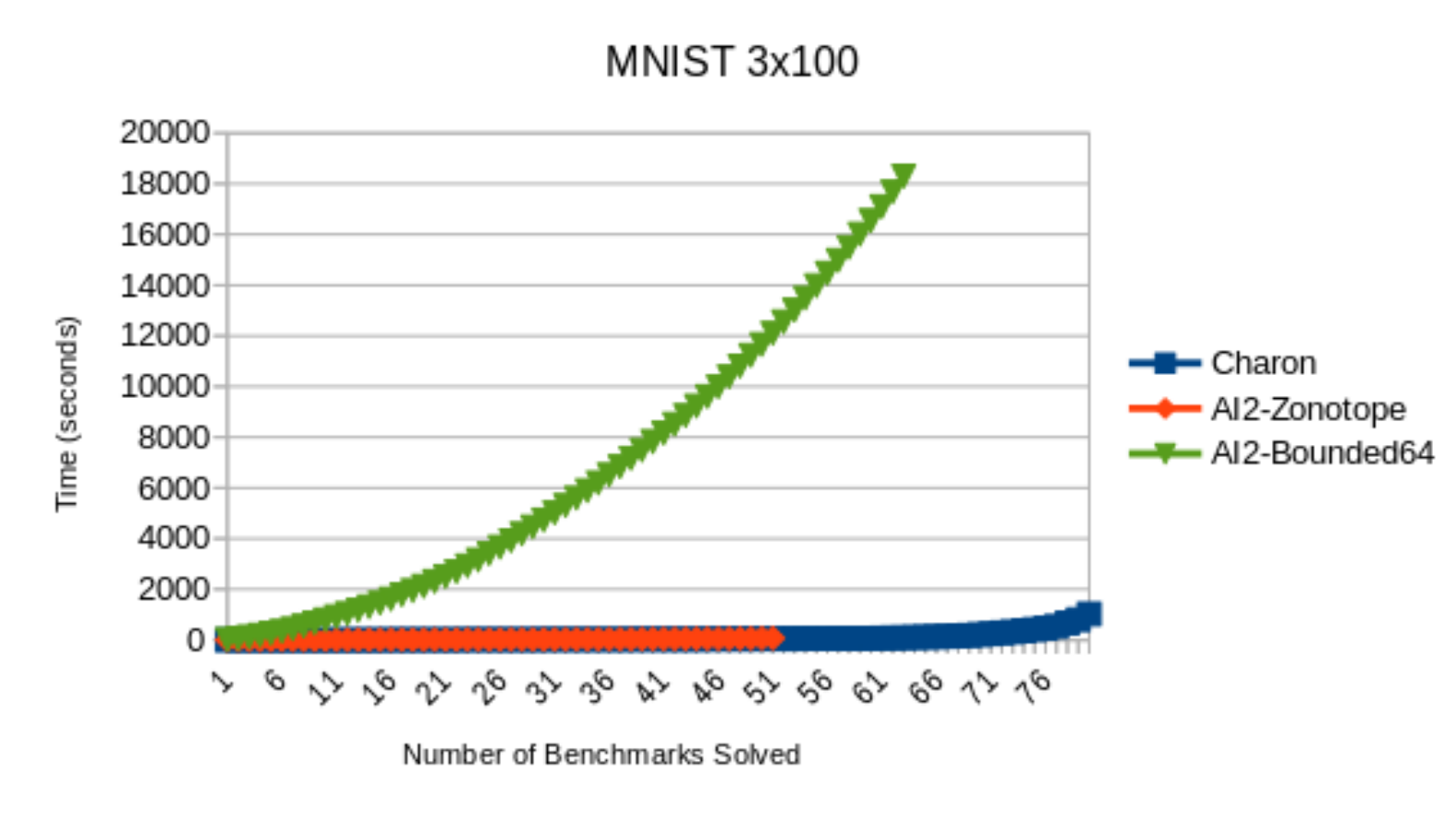}
  \vspace{-0.4in}
  \caption{Comparison on a 3x100 MNIST network.}
  \label{fig:mnist-3-100}
    \vspace{-0.1in}
\end{figure}

\begin{figure}
  \includegraphics[width=\columnwidth]{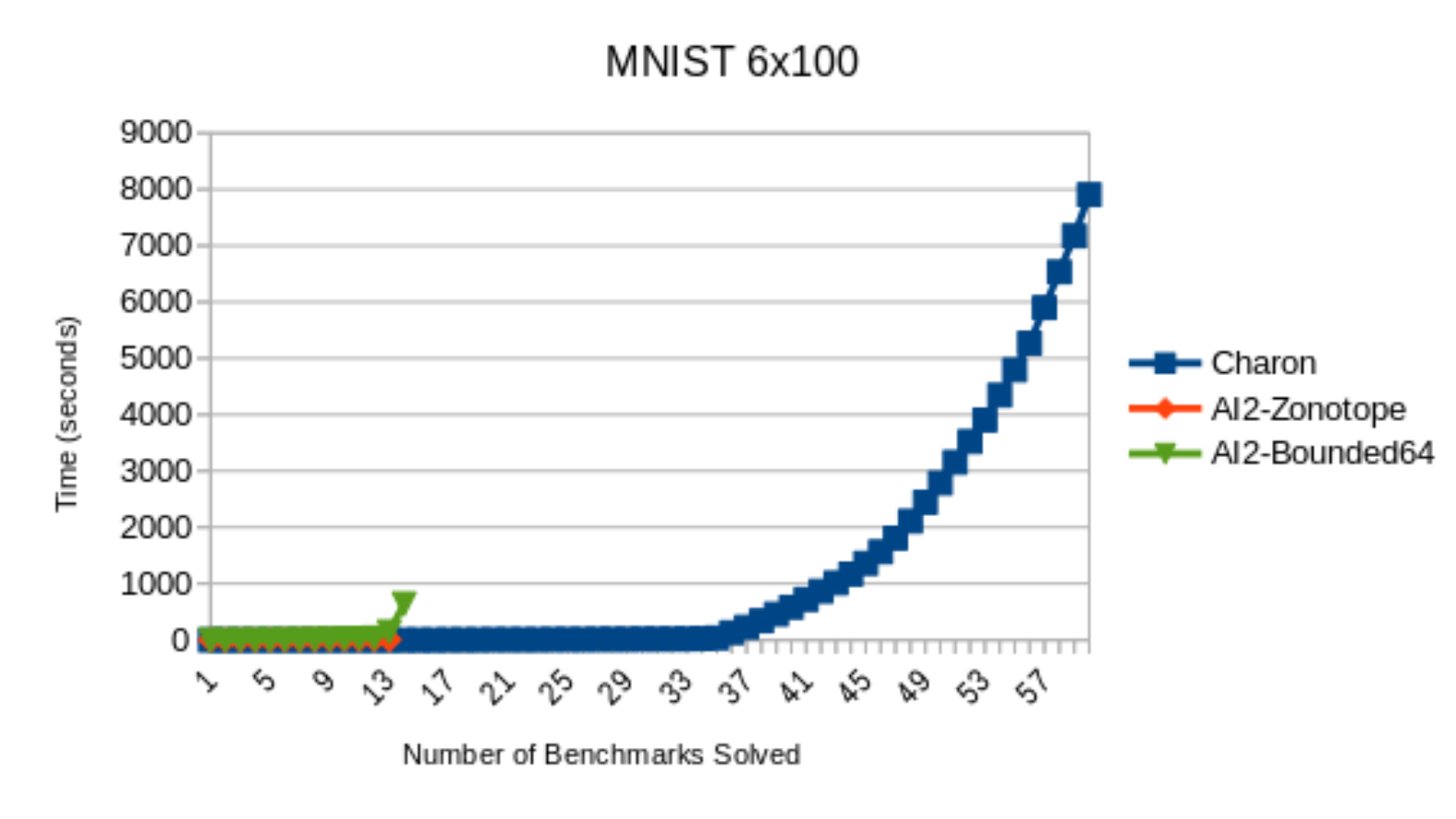}
    \vspace{-0.4in}
  \caption{Comparison on a 6x100 MNIST network.}
  \label{fig:mnist-6-100}
    \vspace{-0.1in}
\end{figure}

\begin{figure}
  \includegraphics[width=\columnwidth]{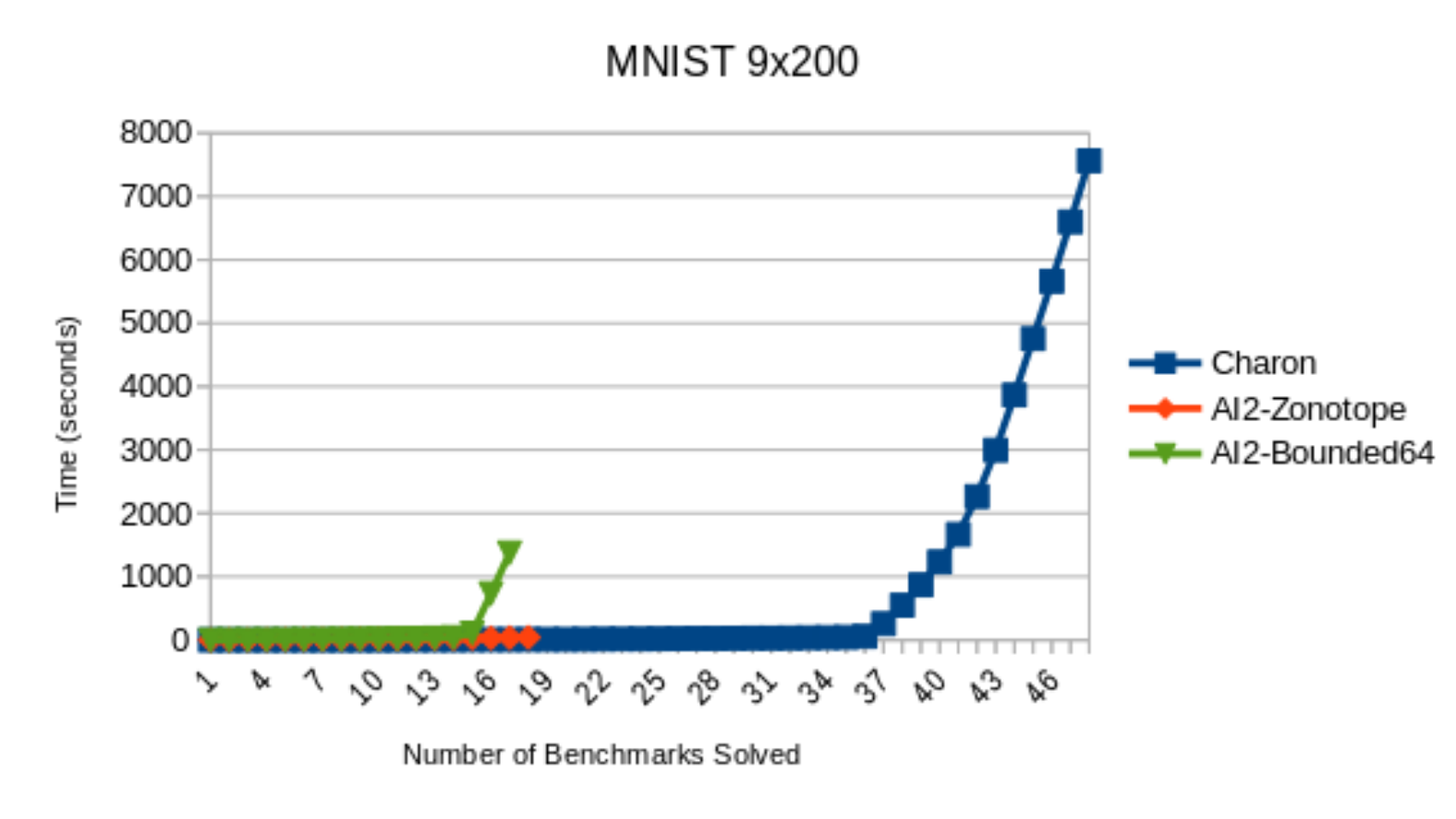}
    \vspace{-0.4in}
  \caption{Comparison on a 9x200 MNIST network.}
  \label{fig:mnist-9-200}
      \vspace{-0.1in}
\end{figure}

\begin{figure}
  \includegraphics[width=\columnwidth]{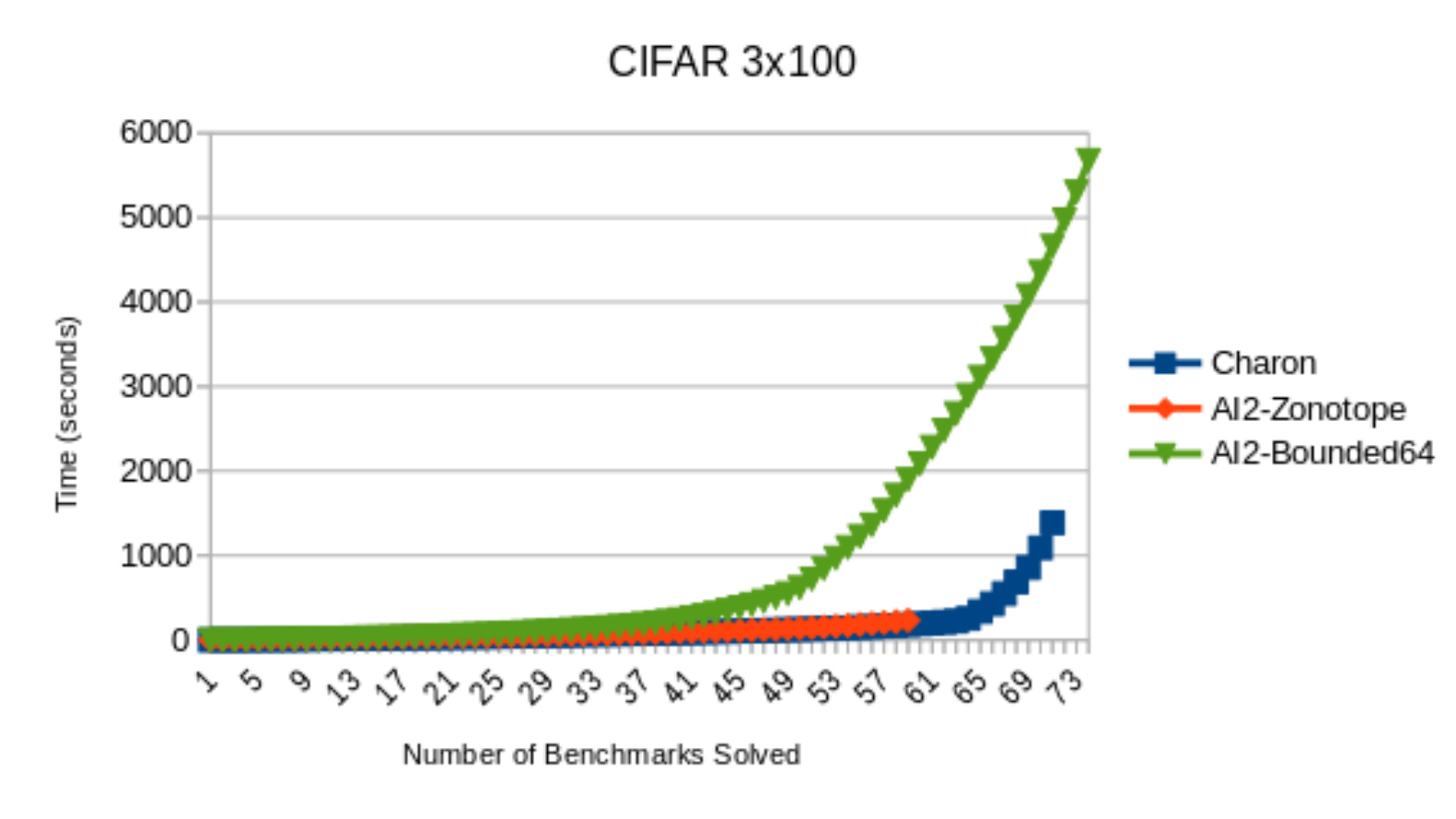}
    \vspace{-0.4in}
  \caption{Comparison on a 3x100 CIFAR network.}
  \label{fig:cifar-3-100}
      \vspace{-0.1in}
\end{figure}

In this section we compare \toolname\ with \ai~\footnote{Because we did not have
access to the original \ai, we reimplemented it. However, to allow for a fair
comparison, we use the same underlying abstract interpretation library, and we
implement the transformers exactly as described in \cite{ai2}.},  a state-of-the-art   tool for verifying network robustness~\cite{ai2}.  As discussed in Section~\ref{sec:background}, \ai\ is incomplete and requires the user to specify which abstract domain to use.  Following their evaluation strategy from the IEEE S\&P paper~\cite{ai2},  we instantiate \ai\ with two different domains,
namely zonotopes and bounded powersets of zonotopes of size 64.  We refer to these two variants as \ai-Zonotope and \ai-Bounded64.

The results of this comparison are summarized in Figure~\ref{fig:ai2-summary}.
This graph shows the percentage of benchmarks each tool was able to verify or
falsify, as well as the percentage of benchmarks where the tool timed out and
the percentage where the tool was unable to conclude either true or false. Note
that, because \toolname\ is $\delta$-complete, there are no ``unknown'' results
for it, and because \ai\ cannot find counterexamples, \ai\ has no ``falsified''
results.

 The details for each network are  shown in Figures~\ref{fig:mnist-3-100}
-~\ref{fig:mnist_conv}. Each chart shows the cumulative time taken on the
y-axis and the number of benchmarks solved on the x-axis (so lower is better). 
The results for each tool
include only those benchmarks  that the tool could solve correctly within the 
time limit of 1000 seconds. Thus, a line extending further to the right indicates that the tool could
solve more benchmarks. Since \ai-Bounded64 times out on every benchmark for
the convolutional network, it does not appear in Figure~\ref{fig:mnist_conv}.

\begin{figure}
  \includegraphics[width=\columnwidth]{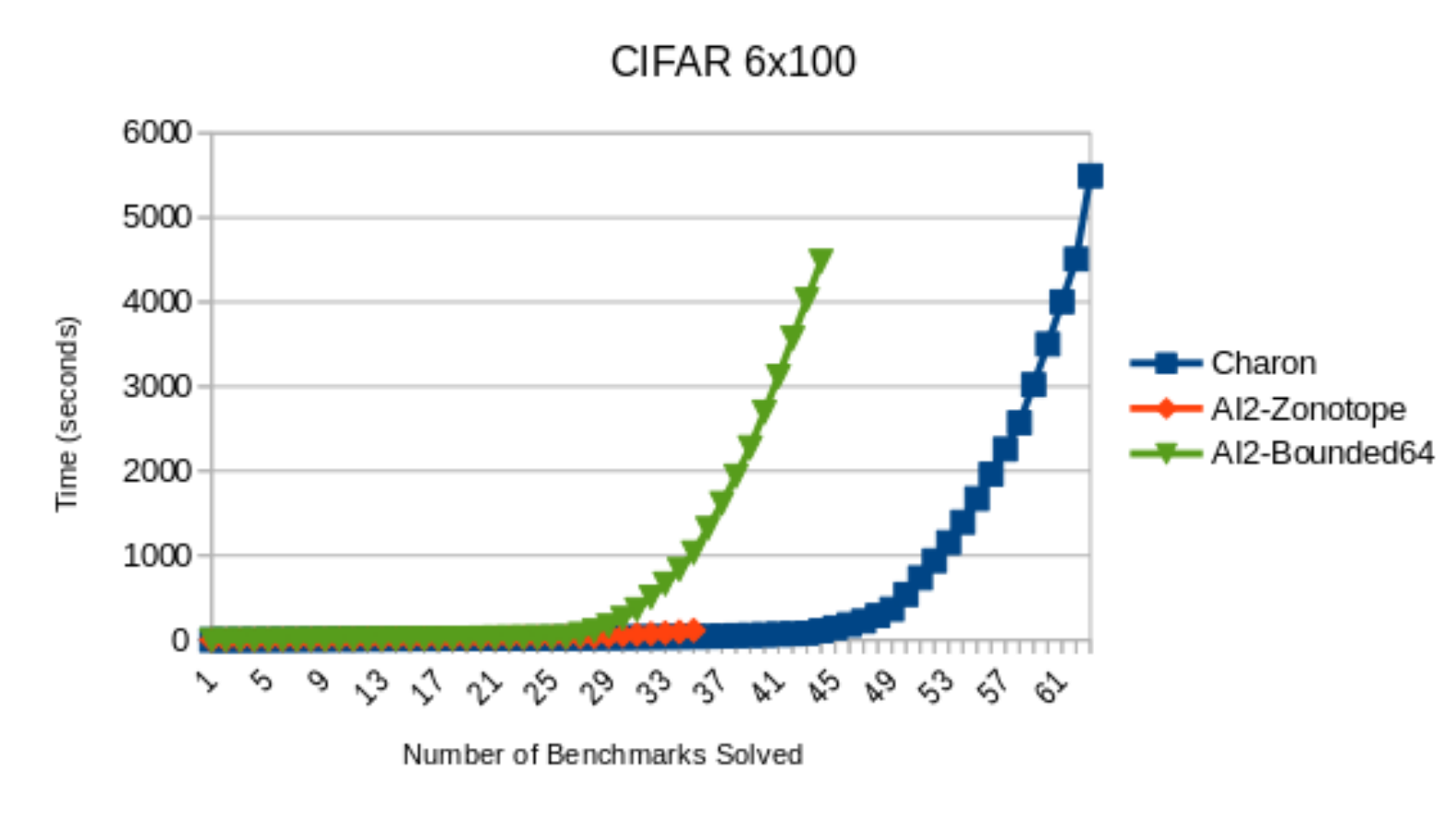}
    \vspace{-0.4in}
  \caption{Comparison on a 6x100 CIFAR network.}
  \label{fig:cifar-6-100}
      \vspace{-0.1in}
\end{figure}

\begin{figure}
  \includegraphics[width=\columnwidth]{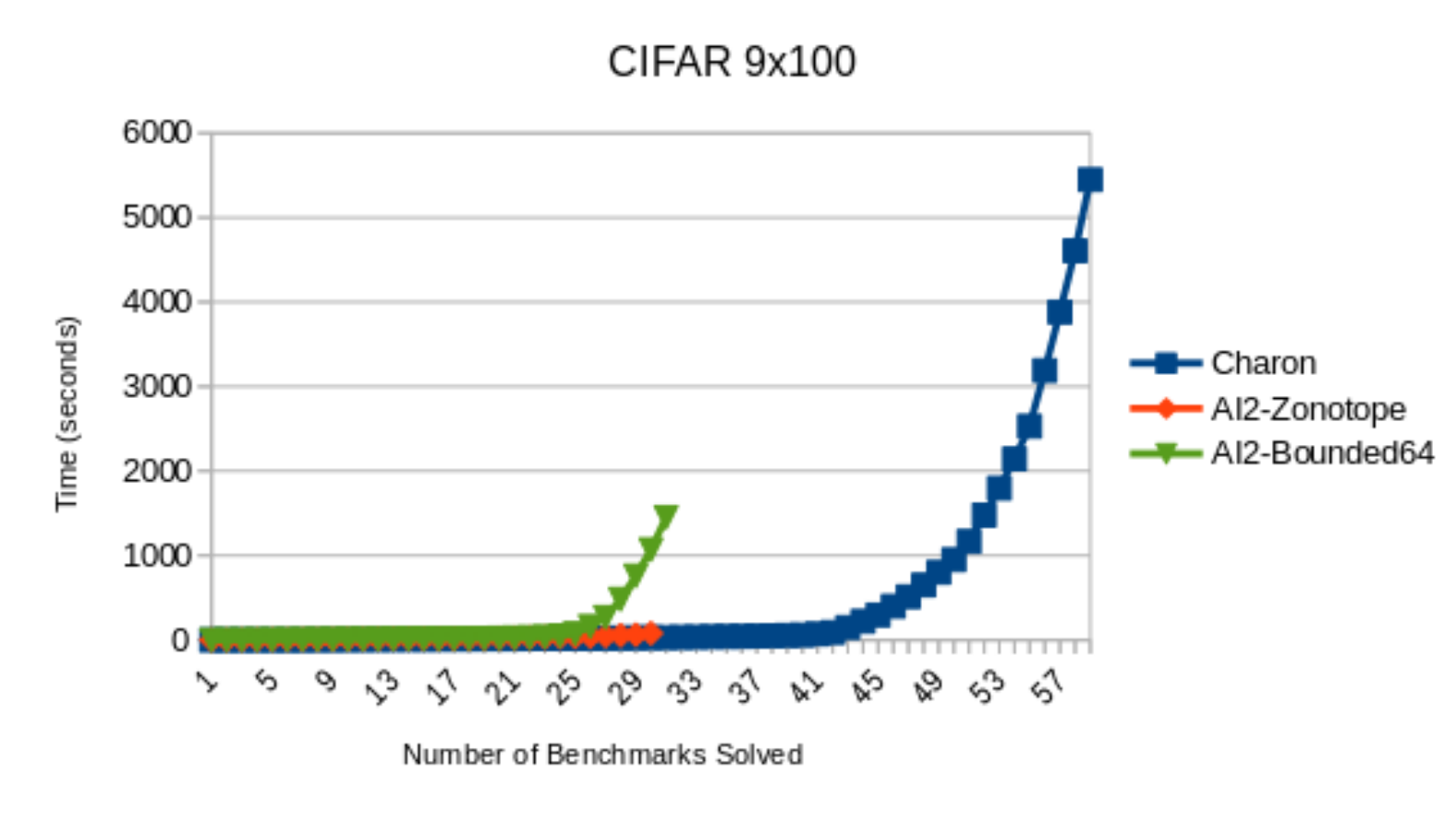}
    \vspace{-0.4in}
  \caption{Comparison on a 9x100 CIFAR network.}
  \label{fig:cifar-9-100}
      \vspace{-0.1in}
\end{figure}

\begin{figure}
  \includegraphics[width=\columnwidth]{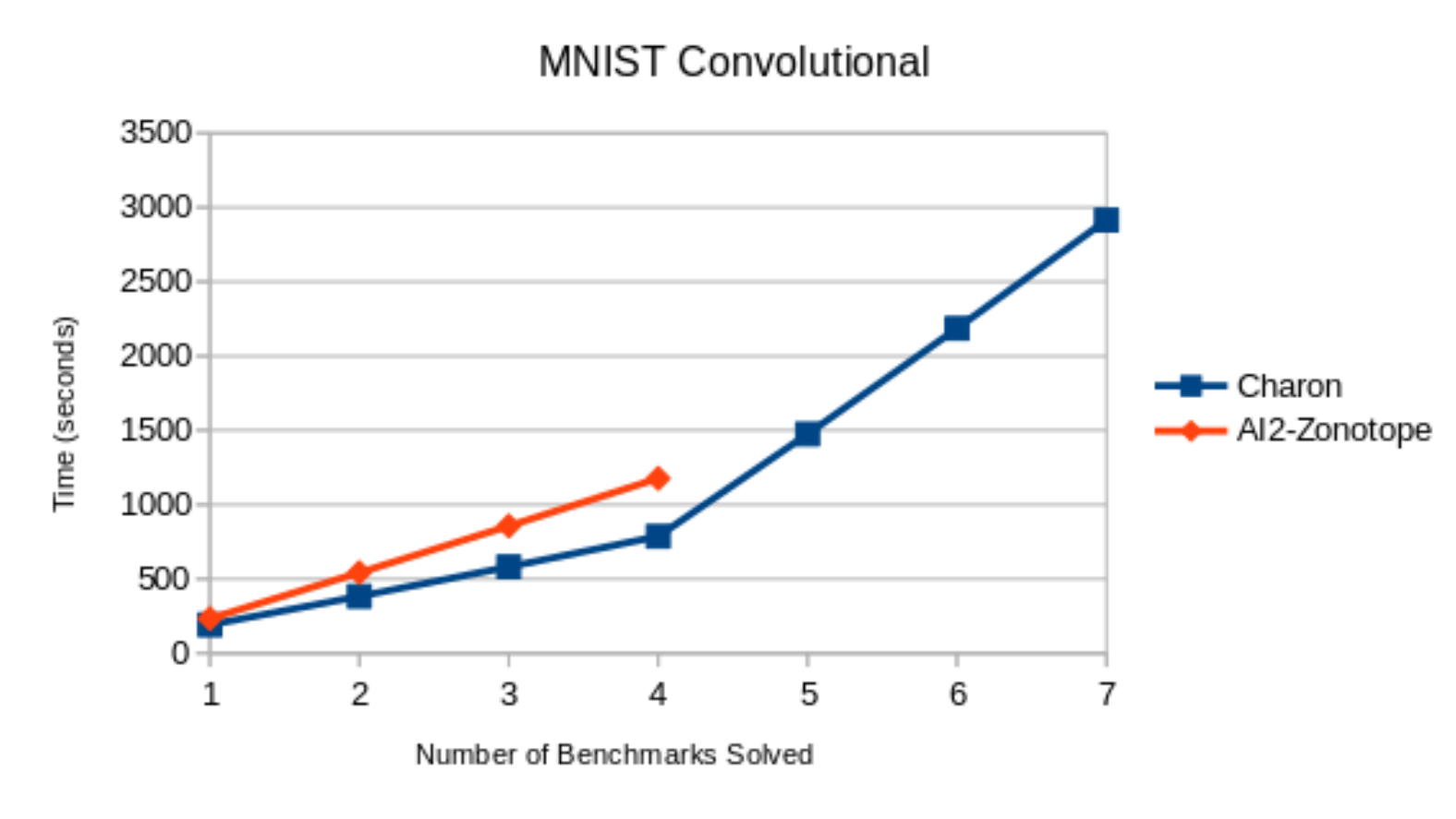}
    \vspace{-0.4in}
  \caption{Comparison on a convolutional network.}
  \label{fig:mnist_conv}
      \vspace{-0.2in}
\end{figure}

The key take-away lesson from this experiment is that \toolname\ is able to both solve
more benchmarks compared to \ai-Bounded64 on most networks, and it is
able to solve them much faster. In particular,  \toolname solves $59.7\%$ (resp.
$84.7\%$) more benchmarks compared to \ai-Bounded64 (resp. \ai-Zonotope).
Furthermore, among the benchmarks that can be solved by both tools, \toolname is
6.15$\times$ (resp. 1.12$\times$ ) faster compared to \ai-Bounded64 (resp. \ai-Zonotope). Thus, we believe these results demonstrate 
the advantages of our approach compared to \ai.


\subsection{ Comparison with Complete Tools (RQ1)}\label{sec:reluval-comp}

In this section we compare \toolname\ with other complete tools for robustness
analysis, namely \reluval~\cite{reluval} and \reluplex~\cite{reluplex}. Among these tools, \reluplex 
implements a variant of Simplex with built-in
support for the ReLU activation function~\cite{reluplex}, and \reluval\ is an
abstraction refinement approach without learning or counterexample search.

To perform this experiment, we evaluate all three tools on the
same  benchmarks from Section~\ref{sec:ai2-compare}. However, since \reluval\ and \reluplex\ do not support convolutional layers,
we exclude the convolutional net  from this evaluation.

\begin{figure}
  \includegraphics[width=\columnwidth]{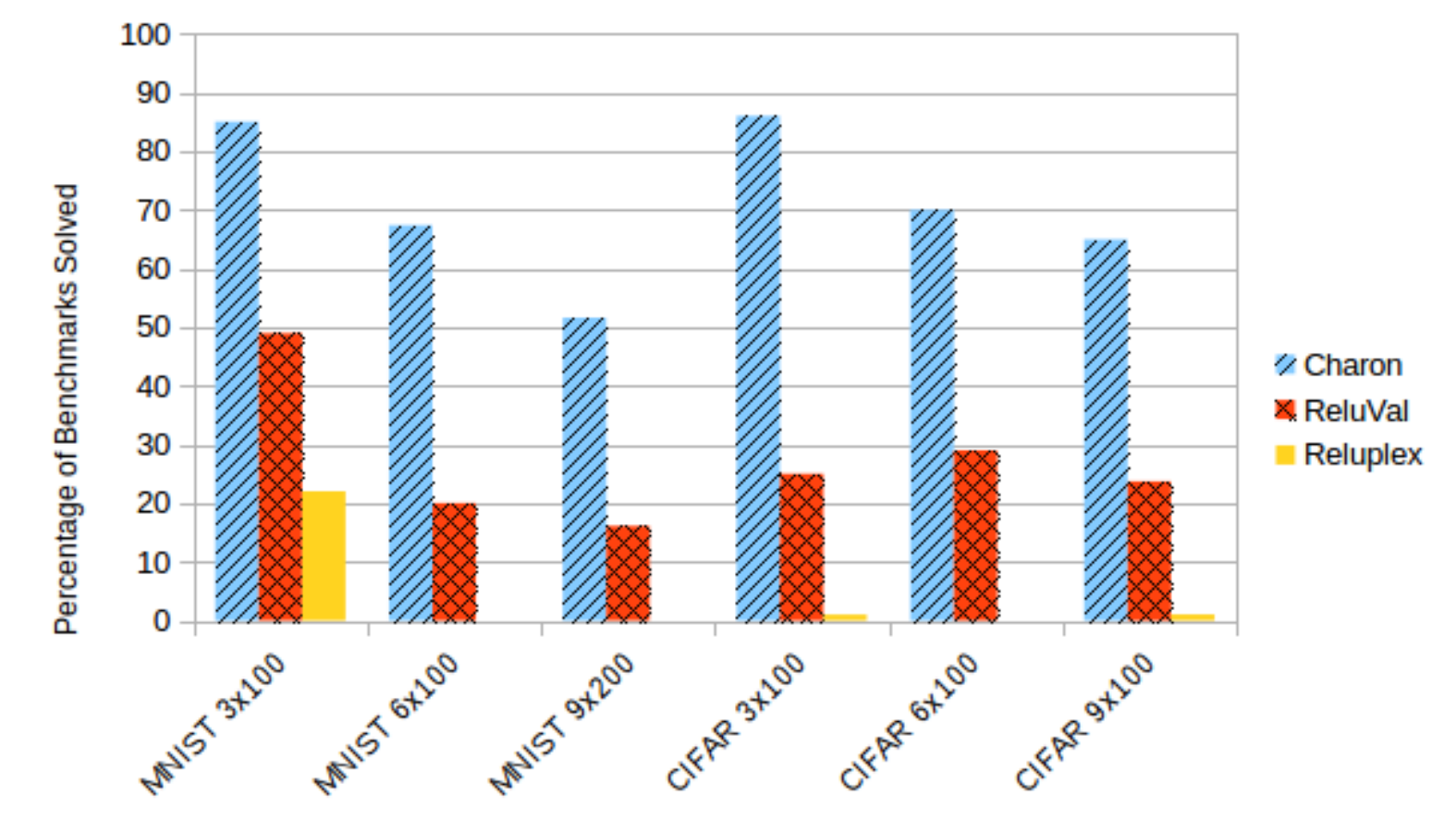}
  \vspace{-0.1in}
  \caption{Comparison with \reluval.}
  \label{fig:reluval}
  \vspace{-0.1in}
\end{figure}

The results of this comparison are summarized in Figure~\ref{fig:reluval}.  Across all benchmarks,
\toolname\ is able to solve 2.6$\times$ (resp. 16.6$\times$)  more problems
compared to \reluval (resp. \reluplex). Furthermore, it is worth noting that the
set of benchmarks that can be solved by \toolname\ is a strict superset of the
benchmarks solved by \reluval.

\subsection{Impact of Counterexample Search (RQ2) }

To understand the benefit of using optimization to search for counterexamples, we now compare the number 
of properties that can be \emph{falsified} using \toolname vs. \reluplex and \reluval. (Recall that \ai\ is incomplete and 
cannot be used for falsification.) Among the 585 benchmarks used in the evaluation from Section~\ref{sec:reluval-comp}, \toolname\ can 
falsify robustness of 123 benchmarks. In contrast, \reluplex can only falsify robustness of one benchmark, and \reluval cannot falsify any of them. Thus, we believe these results demonstrate the usefulness of incorporating optimization-based counterexample search into the decision procedure.

\subsection{ Impact of Learning a Verification Policy (RQ3)}

\begin{figure}
  \includegraphics[width=\columnwidth]{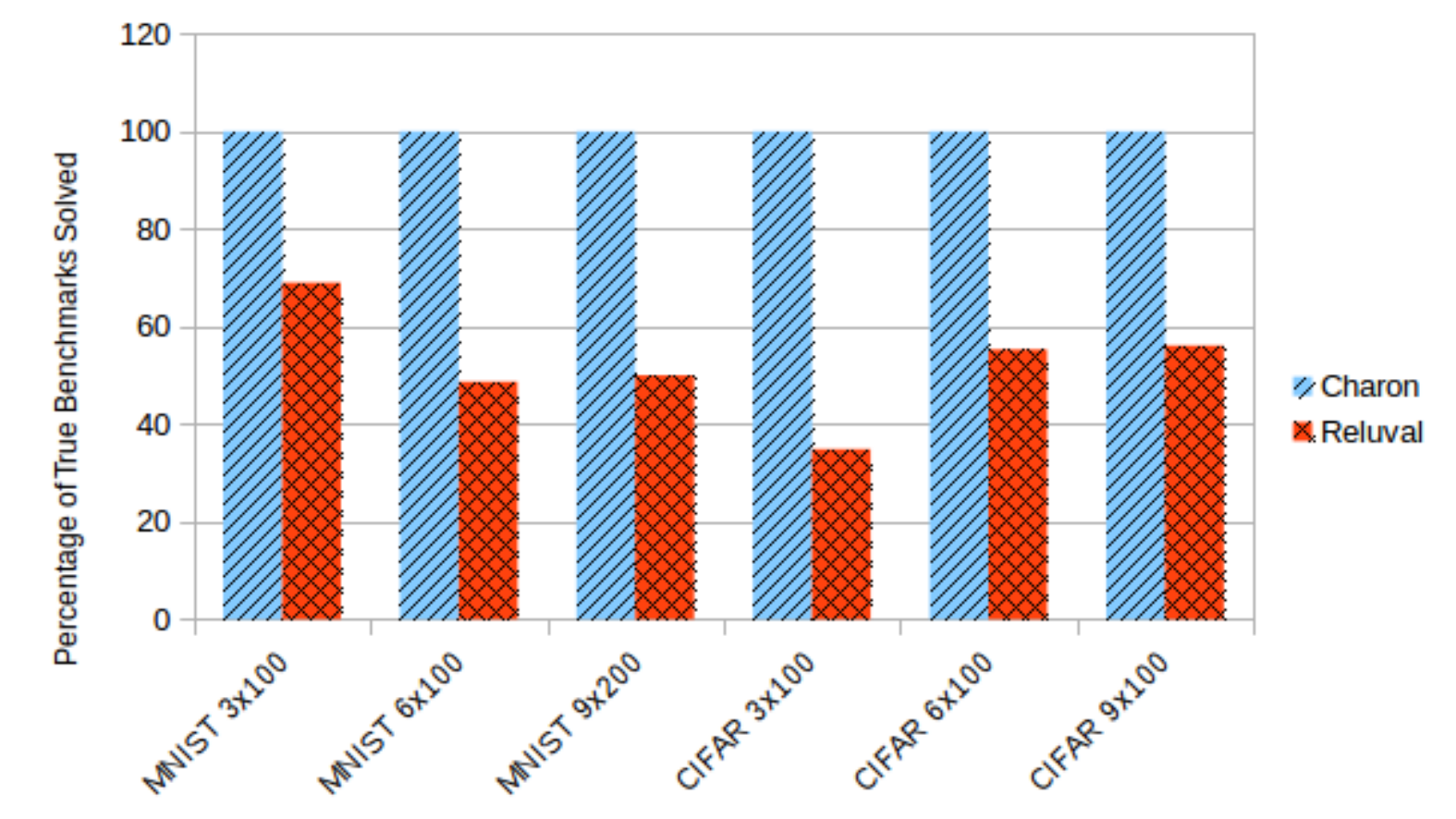}
  \vspace{-0.1in}
  \caption{Comparison with \reluval\ on verified benchmarks.}
  \vspace{-0.2in}
  \label{fig:reluval-verified}
\end{figure}

Recall that a key feature of our algorithm is the use of a machine-learnt verification policy $\policy$
to choose a refinement strategy. To explore the impact of this design choice, we compare our technique against \reluval on 
the subset of the 585 benchmarks for which the robustness property holds. In
particular, as mentioned earlier, \reluval is also based on a form of
abstraction refinement but uses a static, hand-crafted strategy rather than  one
that is learned from data. Thus, comparing against \reluval on the
verifiably-robust benchmarks allows us to evaluate the benefits of learning a
verification policy from data.~\footnote{We compare with \reluval directly
rather than reimplementing the \reluval strategy inside \toolname because our
abstract interpretation engine does not support the domain used by \reluval.
Given this, we believe the comparison to \reluval is the most fair available
option.}

The results of this comparison are shown in Figure~\ref{fig:reluval-verified}. As we can see from this figure, \reluval\ is still only
able to solve between 35-70\% of the benchmarks that can be successfully solved by \toolname. Thus, these results demonstrate that our data-driven approach to learning verification policies is useful for verifying network robustness.

\section{Related Work}
\label{sec:related}

In this section, we survey existing work  on robustness analysis of neural networks 
and other ideas related to this paper.

\vspace{-0.07in}
\paragraph{Adversarial Examples and Robustness.}

Szegedy et al. \cite{SzegedyZSBEGF13} first showed that neural
networks are vulnerable to small perturbations on inputs. 
It has since been shown
that such examples can be exploited to attack machine learning systems in
safety-critical applications such as autonomous robotics~\cite{melis2017deep} and malware
classication~\cite{grosse2017adversarial}.  

Bastani et al.~\cite{BastaniILVNC16} formalized the notion of
 local robustness in neural networks and defined metrics to evaluate the
 robustness of a neural network. Subsequent work has introduced other
 notions of robustness~\cite{reluplex,gopinath2017deepsafe}.

Many recent papers have studied the construction of
adversarial counterexamples
\cite{KurakinAdversarialScale,LyuHL15,NguyenYC15,TabacofV16,SabourCFF15,GrossePM0M16,madry2017towards,
  GoodfellowSS14}.  These approaches are based on various forms of
gradient-based optimization, for example L-BFGS~\cite{szegedy2013intriguing}, FGSM~\cite{GoodfellowSS14} and PGD~\cite{madry2017towards}. While our
implementation uses the PGD method, we could in principle also use (and benefit from advances in) alternative gradient-based optimization methods.

\vspace{-0.07in}
\paragraph{Verification of Neural Networks.} 

Scheibler et
al. \cite{ScheiblerWWB15} used bounded model checking to verify safety
of neural networks. 
Katz et al.\cite{KatzBDJK17} developed the
Reluplex decision procedure extending the Simplex algorithm to verify
robustness and safety properties of feedforward networks with ReLU units.
Huang et al.\cite{HuangKWW16}~showed a
verification framework, based on an SMT solver, which verified
robustness with respect to a certain set of functions that can
manipulate the input.  
A few recent papers~\cite{TjengMILP, DuttaMILP, LomuscioMILP} use Mixed
Integer Linear Programming (MILP) solvers to verify local robustness
properties of neural networks. 
These methods do not use abstraction and do not scale very well, but combining
these techniques with abstraction is an interesting
area of future work.

The earliest effort on neural network verification to use abstraction
was by Pulina and Tacchella~\cite{PulinaT10} --- in fact,  like our method,
they considered an abstraction-refinement approach to solve this problem. However, their approach
represents abstractions using 
general linear arithmetic formulas and uses a decision procedure to
perform verification and counterexample search. Their approach 
was shown to be successful for a network with only $6$
neurons, so it does not have good scalability properties.  More recently, Gehr
et al. \cite{ai2} presented the AI$^2$ system for abstract
interpretation of neural networks. Unlike our work, AI$^2$ is
incomplete and cannot produce concrete counterexamples. The most
closely related approach from prior work is \reluval~\cite{reluval},
which performs abstract
interpretation using symbolic intervals. The two key differences between
\reluval and our work are that \toolname couples abstract interpretation
with optimization-based counterexample search and learns
verification policies from data. As demonstrated in Section~\ref{sec:eval}, both of 
these ideas have a significant impact on our empirical results.




\vspace{-0.07in}
\paragraph{Learning to Verify.}

The use of data-driven learning in neural network
verification is, so far as we know, 
new.
However, there are many papers~\cite{sharma2013verification,liang2011learning,garg2014ice,heo2016learning,si2018learning}
on the use of such learning in traditional software verification. While most of these efforts learn proofs from execution data for specific
programs, there are a few efforts that seek to learn optimal
instantiations of parameterized abstract domains from a corpus of training {problems}~\cite{liang2011learning,oh2015learning}. The most
relevant work in this space is by Oh et al.~\cite{oh2015learning},
who use Bayesian optimization to adapt a parameterized abstract
domain. The abstract domain in that work is finite, and the Bayesian
optimizer is only used to adjust the context-sensitivity and
flow-sensitivity of the analysis. In contrast, our analysis of neural
networks handles real-valued data and a possibly infinite space of strategies.

\section{Conclusion and Future Work}\label{sec:conclusion}
We have presented a novel technique for verifying robustness properties of
neural networks based on synergistically combining proof search with
counterexample search. This technique makes use of black-box optimization
techniques in order to learn good refinement strategies in a data-driven way. We
implemented our technique and showed that it significantly outperforms
state-of-the-art techniques for robustness verification. Specifically, our
technique is able to solve 2.6$\times$ as many benchmarks as \reluval and
16.6$\times$ as many as \reluplex. Our technique is able to solve more
benchmarks in general than \ai, and is able to solve them far faster.
Moreover, we provide theoretical
guarantees about the termination and ($\delta$-)completeness of our approach.

In order to improve our tool in the future, we plan to explore a broader
set of abstract domains and different black-box optimization techniques.
Notably, one can view solver-based techniques as a perfectly precise abstract
domain. While solver-based techniques have so far proven to be slow on many
benchmarks, our method could learn when it is best to apply solvers and when to
choose a less precise domain. This would allow the tool to combine solvers and
traditional numerical domains in the most efficient way. Additionally, while
Bayesian optimization fits our current framework well, it may be possible
to modify the framework to work with different learning techniques which can
explore higher-dimensional strategy spaces. In particular, reinforcement
learning may be an interesting approach to explore in future work.

\begin{acks}                            
  We thank our shepherd Michael Pradel as well as our anonymous reviewers and
  members of the UToPiA group for
  their helpful feedback.
  This material is based upon work supported by the
  \grantsponsor{SP784}{National Science
    Foundation}{https://www.nsf.gov/} under Grants
  No.~\grantnum{SP784}{CCF-1162076},
  No.~\grantnum{SP784}{CCF-1704883},
  No.~\grantnum{SP784}{CNS-1646522}, and
  No.~\grantnum{SP784}{CCF-1453386}. This work is also supported by Google
  Cloud.
\end{acks}

\ifextended
\else
\balance
\fi
\bibliography{main}

\ifextended
\appendix

\section{Proofs}
In this section we present the proofs of the theorems in
Section~\ref{sec:theorems}. For convenience, the assumptions and theorem
statements have been copied.

\begin{repdefinition}{def:delta-complete}
For a given network $\network$, input region $I$, target class $M$, and
$\delta>0$, a $\delta$-counterexample is a point $x \in I$ such that for some
$j$ with $1 \le j \le m$ and $j \neq M$, $\network(x)_M - \network(x)_j \le
\delta$.
\end{repdefinition}

\begin{repdefinition}{def:diameter}
  For any set $X \subseteq \real^n$, its \emph{diameter} $D(X)$ is defined as
  \[D(X) = \sup\{\|x_1 - x_2\|_2 \mid x_1,x_2 \in X\}\]
  if this value exists. Otherwise the set is said to have infinite diameter.
\end{repdefinition}

\begin{repassumption}{asm:split-reduces-size}
  There exists some $\lambda \in \real$ with $0 < \lambda
  < 1$ such that for any network $\network$, input region $I$, and point $x_* \in
  I$, if $(I_1, I_2) = \Call{Refine}{N, I, x_*}$, then $D(I_1) < \lambda
  D(I)$ and $D(I_2) < \lambda D(I)$.
\end{repassumption}

\begin{repassumption}{asm:finite-overapprox}
  Let $\network^\#$ be the abstract transformer representing a network
  $\network$. Let $a$ be an element of the abstract domain representing the
  input region $I$.  We assume there exists some $K_\network \in \real$ such
  that $D(\gamma(\network^\#(\alpha(I)))) < K_\network D(I)$.
\end{repassumption}

\begin{reptheorem}{thm:termination}
Consider the variant of Algorithm~\ref{alg:refinement} where the predicate at line 3 is replaced with
Eq.~\ref{eq:modification}. Then, the verification algorithm always terminates
under Assumptions~\ref{asm:split-reduces-size} and~\ref{asm:finite-overapprox}.
\end{reptheorem}
\begin{proof}
  To improve readibility, we define $F(I_k) = \gamma(\network^\#(\alpha(I_k)))$.

  By Assumption~\ref{asm:split-reduces-size} there exists some $\lambda \in
  \real$ with $0 < \lambda < 1$ such that for any input region $I'$, splitting
  $I'$ with \Call{Refine}{} yields regions $I_1'$ and $I_2'$ with $D(I_1') < \lambda D(I')$
  and $D(I_2') < \lambda D(I')$. Because there is one split for each node in the
  recursion tree, at a recursion depth of $k$, the region $I_k$ under
  consideration has diameter $D(I_k) < \lambda^k D(I)$. By
  Assumption~\ref{asm:finite-overapprox}, there exists some $K_\network$ such
  that $D(F(I_k)) < K_\network D(I_k)$. Notice that when
  \[k > \log_\lambda\left(\frac{\delta}{2 K_\network D(I)}\right)\]
  we must have $D(F(I_k)) < \delta / 2$.

  We will now show that when $k$ satisfies the preceding condition,
  Algorithm~\ref{alg:refinement} must terminate without recurring. In this case, suppose $x_*$ is the point returned by the call to
  \Call{Minimize}{} and the algorithm does not terminate. Then
  $\objective(x_*)>\delta$ and in particular, $\network(x_*)_K-\network(x_*)_i >
  \delta$. Since \Call{Analyze}{} is sound, we must have $\network(x_*) \in
  F(I_k)$. Then since $D(F(I_k)) < \delta / 2$, we must have that for any point
  $y' \in F(I_k)$, $\|\network(x_*) - y'\|_2 < \delta / 2$. In particular, for
  all $i$, $|(\network(x_*))_i - y'_i| < \delta / 2$, so $y'_i >
  (\network(x_*))_i - \delta / 2$ and $y'_i < (\network(x_*))_i + \delta / 2$.
  Then, for all $i$,
  \begin{align*}
    y_K - y_i &> ((\network(x_*))_K - \delta / 2) - (\network(x_*))_i + \delta /
    2) \\
    &= ((\network(x_*))_K - (\network(x_*))_i) - \delta \\
    &> 0
  \end{align*}
  Thus, for each point $y' \in F(I_k)$, we have $y'_K > y'_i$. Since $y'$ ranges
  over the \emph{overapproximated} output produced by the abstract interpreter,
  this exactly satisfies the condition which \Call{Analyze}{} is checking, so
  \Call{Analyze}{} must return {\rm Verified}. Therefore, the maximum recursion
  depth of Algorithm~\ref{alg:refinement} is bounded, so it must terminate.
\end{proof}

\begin{reptheorem}{thm:completeness}
Consider the variant of Algorithm~\ref{alg:refinement} where the predicate at line 3 is replaced with
Eq.~\ref{eq:modification}. Then, the verification algorithm is
$\delta$-complete, meaning that if the algorithm does not return ``Verified''
then the return value is a {$\delta$-counterexample} for the property.
\end{reptheorem}
\begin{proof}
  First note that by Theorem~\ref{thm:termination},
  Algorithm~\ref{alg:refinement} must terminate. Therefore the algorithm must
  return some value, and we can prove this theorem by analyzing the possible
  return values.
  There are five places at which Algorithm~\ref{alg:refinement} can terminate:
  lines 4, 8, 12, 15, and 16. We only care about the case where the algorithm does
  not return ``Verified'' so we can ignore lines 8 and 16. The return at line 4
  is only reached after checking that $x_*$ is a {$\delta$-counterexample}, so
  clearly if that return statement is used then the algorithm returns a
  {$\delta$-counterexample}. This leaves the return statements at lines 12 and
  15. We suppose by induction that the recursive calls at lines 10 and 13 are
  $\delta$-complete. Then if $r_1$ is not {\rm Verified}, it must be a
  {$\delta$-counterexample}. Thus the return statement at line 12 also returns a
  {$\delta$-counterexample}. Similarly, if $r_2$ is not {\rm Verified}, then it
  is a {$\delta$-counterexample}, so line 15 returns a
  {$\delta$-counterexample}.
\end{proof}

\fi

\end{document}